\documentclass[11pt]{article}
\usepackage[margin=.95in]{geometry}
\usepackage{cite}
\usepackage{amsmath,amsthm,amssymb,amsfonts}
\usepackage{algorithm,algorithmic}
\usepackage{graphicx}
\usepackage[table]{xcolor}
\usepackage{hyperref}
\usepackage{authblk}
\usepackage{pdfpages}
\usepackage{booktabs}
\usepackage{caption}
\usepackage{subcaption}
\usepackage{multirow}
\usepackage{mathtools} 
\def\BibTeX{{\rm B\kern-.05em{\sc i\kern-.025em b}\kern-.08em
    T\kern-.1667em\lower.7ex\hbox{E}\kern-.125emX}}

\newtheorem{thm}{Theorem}
\newtheorem{definition}{Definition}
\newtheorem{problem}{Problem}

\newtheorem{lemma}{Lemma}

\newcommand{\done}[1]{\textcolor{green}{(done)}}
\newcommand{\ip}[1]{\textcolor{orange}{(in progress)}}

\newcommand{\real}{\ensuremath{\rm I\!R}}

\newcommand{\etal}{\emph{et. al}}
\newcommand{\lxor}{\oplus}

\newcommand{\infset}{\ensuremath{\mathcal{I}}}
\newcommand{\compset}{\ensuremath{\mathcal{C}}}
\newcommand{\accset}{\ensuremath{\mathcal{A}}}
\newcommand{\probmat}{\ensuremath{M}}
\newcommand{\inputspace}{{\ensuremath{\mathcal{X}^n}}}

\newcommand{\outputspace}{{\ensuremath{\mathcal{Y}}}}
\newcommand{\likelyratio}[1]{\ensuremath{\frac{\pr{A(\mathbf{x}_{#1}) = y}}{\pr{A(\mathbf{x}_{#1}') = y_{#1}}}}}
\newcommand{\outputspacerr}{{\ensuremath{\mathcal{Y}^n}}}
\newcommand{\targetvals}{\ensuremath{\mathcal{V}}}
\newcommand{\pr}[1]{\ensuremath{Pr({#1})}}
\newcommand{\wmctime}{\ensuremath{\mathsf{WMC}}}
\newcommand{\wmc}{\ensuremath{WMC}}
\newcommand{\rr}{\ensuremath{\mathsf{RR}}}
\newcommand{\rrcount}{\ensuremath{\mathsf{RRcount}}}

\newif\ifanonymous
\newif\ifdraft
\drafttrue

\title{Synthesizing Tight Privacy and Accuracy Bounds via Weighted Model Counting
}

\ifanonymous
\else
\author{Lisa Oakley \qquad Steven Holtzen \qquad Alina Oprea}
\affil{ \small \emph{Khoury College of Computer Sciences, Northeastern University }}
\fi

\date{\vspace{-5ex}}

\begin{document}
\maketitle

\begin{abstract}
    Programmatically generating tight differential privacy (DP) bounds is a hard problem. Two core challenges are (1) finding expressive, compact, and efficient encodings of the distributions of DP algorithms, and (2) state space explosion stemming from the multiple quantifiers and relational properties of the DP definition.
    
    We address the first challenge by developing a method for tight privacy and accuracy bound synthesis using weighted model counting on binary decision diagrams, a state of the art technique from the artificial intelligence and automated reasoning communities for exactly computing probability distributions. We address the second challenge by developing a framework for leveraging inherent symmetries in DP algorithms. Our solution benefits from ongoing research in probabilistic programming languages, allowing us to succinctly and expressively represent different DP algorithms with approachable language syntax that can be used by non-experts.

    We provide a detailed case study of our solution on the binary randomized response algorithm. We also evaluate an implementation of our solution using the Dice probabilistic programming language for the randomized response and truncated geometric above threshold algorithms. We compare to prior work on exact DP verification using Markov chain probabilistic model checking. Very few existing works consider mechanized analysis of accuracy guarantees for DP algorithms. We additionally provide a detailed analysis using our technique for finding tight accuracy bounds for DP algorithms.
\end{abstract}

\section{Introduction}
\label{sec:intro}

Differential privacy (DP)~\cite{DworkMNS16} is an important property for randomized algorithms that can ensure a balance between preserving user privacy and allowing systems to draw meaningful conclusions from their data. Crafting algorithms which satisfy meaningful differential privacy bounds while maintaining useful accuracy guarantees is no small feat, and the design and analysis of these differentially private algorithms is extensive and highly technical. As is often the case with complicated, technical fields, as the landscape of differential privacy research has grown, so too has the tendency for bugs in the theory and implementations of these algorithms \cite{mironov2012significance,tramer2022debugging,stevens2022backpropagation}.

It is therefore vital for algorithm designers to have tools and frameworks to help formally and mechanically analyze their algorithms both in the design phase, and to validate their theoretical results and implementations. Furthermore, it is important that these tools be as accessible and automated as possible.

Most current methods for formally verifying differential privacy properties involve manually mechanizing existing pen-and-paper privacy proofs using proof assistants like EasyCrypt \cite{barthe2013verified}. This method of deductive verification is useful for mechanically validating asymptotic privacy bounds \cite{barthe2013verified,barthe2012probabilistic,barthe2014proving,barthe2016proving,hsu2017probabilistic}. This method, however, requires almost entirely manual work and often hinges on first having an existing proof written by hand. Importantly, it also requires a technician who understands both theoretical differential privacy and proof assistants well enough to translate the pen-and-paper proof into the language of the proof assistant. Therefore, these methods can be extremely useful, but are not practical for keeping up with the pace of ongoing differential privacy research, and are not as helpful in the algorithmic design phase where there is incomplete pen-and-paper analysis.
\begin{figure*}[]
    \includegraphics[width=\linewidth,trim={0 7cm 0 7cm},clip]{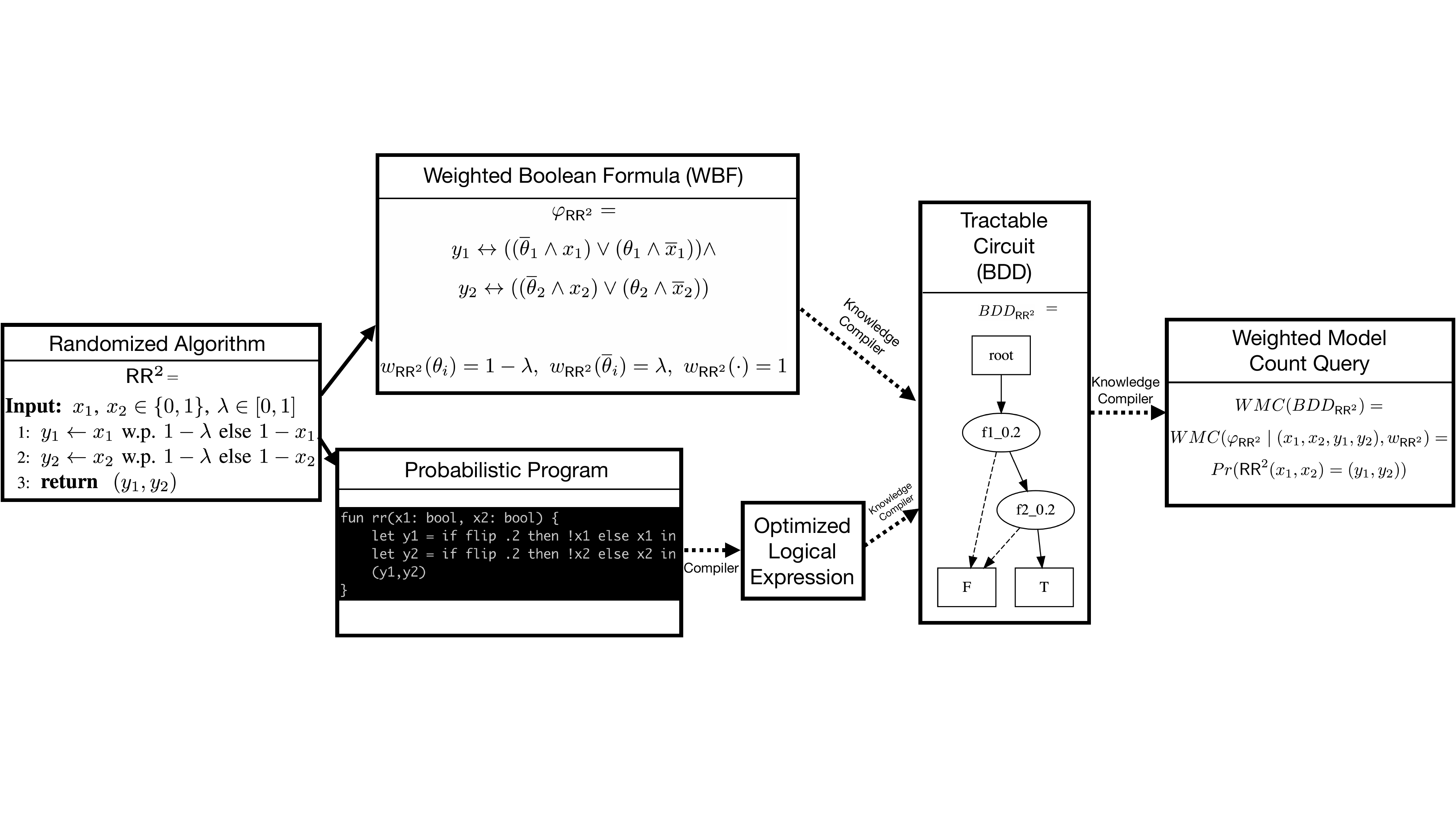}
    \caption{We represent the probability distribution of a randomized algorithm as a binary decision diagram (BDD) so that we can efficiently compute the weighted model count to perform probabilistic inference (compute the probabilities of certain events occurring). Starting from a randomized algorithm, there are many ways to compile to a BDD. The top path shows a more manual process where a user crafts a weighted boolean formula by hand and uses a \emph{knowledge compiler}, a tool for efficiently representing and querying logical formulas, to compile this into a BDD. The bottom path exhibits an example of the flow of an expressive probabilistic programming languages which allow the user to define an easily readable program defined similarly to the pseudocode. From here, the programming language can compile an optimized logical expression and pass it to the knowledge compiler. In this way, there is very little manual effort, and the resulting logical expression and BDD are optimized for WMC queries. }
    \label{fig:diagram}
\end{figure*}

There is also active research on statistically verifying differential privacy in practice, discovering lower bounds by finding counterexamples to differential privacy, or programmatically tightening existing bounds \cite{ding2018detecting,wang2020checkdp,bichsel2018dp,zhang2020testing}. In addition, privacy auditing provides a set of statistical techniques for estimating  privacy leakage of machine learning (ML) algorithms empirically and determining lower bounds on privacy~\cite{AuditingDP,nasr2021adversary,andrew2023oneshot,nasr2023tight,pillutla2023unleashing,steinke2023auditing}. While these methods are useful for automatically analyzing algorithms over large data sets, they often rely on statistical procedures and are evaluated mostly through experimental methods on specific data sets.

Existing methods here also focus almost entirely on verifying privacy, and do not analyze accuracy bounds. This is a problem because incorrect or non-existent accuracy bounds forces a reliance on proxy functions for accuracy that can be flawed and result in algorithm designers adding too much noise, or adding noise in non-optimal parts of the algorithm. Furthermore, without full accuracy analysis, applications like optimization-based synthesis for differentially private algorithms \cite{roy2021learning} must use these flawed accuracy proxy functions and end up finding potentially non-optimal solutions. 

Our goal is to develop a technique for exactly solving the tight differential privacy bound and tight accuracy bound synthesis problems. In other words, given a randomized algorithm, find the parameters which tightly bound its privacy and accuracy guarantees. Our technique will provide more theoretical grounding to the analysis of differentially private mechanisms than the statistical methods, and provide more automation than the traditional proof mechanizing process described in prior work. It will also critically fill a gap in mechanized analysis for accuracy bounds. 

Beyond its utility in finding tight privacy and accuracy bounds, another benefit of our technique is that it allows us to identify which input, neighbor, and output assignments lead to these bounds. In other words, we can identify the worst-case assignments. We provide empirical examples of using our technique to find the top worst-case assignments for accuracy, which highlights how our framework can help an algorithm designer determine which outliers are most impactful on privacy and accuracy bounds. 

Very few attempts have been made at exactly synthesizing tight privacy bounds, and those that exist are extremely limited in the kinds of algorithms they can analyze \cite{liu2018model,liu2022verifying}. As far a we are aware, no one has attempted to mechanically solve for exact accuracy bounds of a given differentially private algorithm. 

There are two main reasons this is the case. Firstly, finding an easily computable encoding of the algorithm on which to do probabilistic inference is hard. Previous work on exact DP verification has required an explicit, hand crafted Markov chain as input \cite{liu2018model,liu2022verifying}. Secondly, exact verification by quantifying over all neighboring inputs and outputs is intractable as the data sets increase in size. Our approach addresses both of these key challenges, and the resulting technique improves the feasibility of exactly solving the privacy and accuracy bound synthesis problems.

To tackle the first challenge, we use a method called weighted model counting (WMC) for computing the probabilities of certain events (otherwise known as performing \textit{probabilistic inference}). WMC is an established state-of-the-art strategy for performing discrete probabilistic inference in the artificial intelligence and automated reasoning community~\cite{chavira2008probabilistic, holtzen2020scaling}. In order to leverage WMC, one must encode our problem into a \emph{weighted Boolean formula} (WBF). We reduce the problem of computing accuracy and privacy bounds of randomized algorithms to computing the WMC of a particular Boolean formula. This technique has many benefits. Firstly, we can rely on the many decades of algorithmic advances in WMC: for instance, high-performance data-structures like binary decision diagrams (BDDs) can efficiently represent the distributions of the randomized algorithms and efficiently query for exact probabilities. Secondly, there exist expressive probabilistic programming languages that can efficiently compile into these optimized representations~\cite{holtzen2020generating}. This makes it simple to encode DP algorithms which often invoke probability distributions, and allows for a usable strategy even for non-experts in probabilistic programming. In Fig. \ref{fig:diagram}, we present an example of this pipeline for randomized response with two inputs.

To tackle the second challenge of state space explosion due to multiple for-all quantifiers on the input and output space, we introduce the concepts of inference, privacy, and accuracy sets, and provide algorithms for synthesizing bounds with respect to these restricted state spaces. These sets allow us to define which inputs and outputs are sufficient for finding privacy and accuracy bounds of DP algorithms. These sets provide a simple framework for utilizing inherent symmetries in DP algorithms to reduce the space of the exact synthesis problems. Despite being simple to define, a limitation of our approach is that there is some manual work which goes into finding these symmetry sets. We provide a detailed analysis of an example of soundly defining the symmetry sets for a randomized response case study. 

Our contributions are (1) a technique for using probabilistic programming languages and WMC for efficient probabilistic inference for DP algorithms, (2) a framework for leveraging symmetries to combat the state space explosion problem for both the tight privacy and accuracy bound synthesis problems, (3) a detailed theoretical analysis of the symmetry sets for a randomized response case study, (4) examples of using our framework for automated analysis of the randomized response and geometric above threshold algorithms, and (5) comparison to existing technique of probabilistic model checking using Markov chains. 

The paper is organized as follows. In Section \ref{sec:prelims} we provide relevant background. In Section \ref{sec:wmc_for_synthesis} we introduce our tight privacy and accuracy bound synthesis problems and explain the benefits of efficient WMC. In Section \ref{sec:symmetries} we introduce a framework for leveraging symmetries. In Section \ref{sec:rr} we provide a detailed case study using our framework for randomized response. In Section \ref{sec:eval} we present experimental results on implementations of our solution using the Dice probabilistic programming language. In Section \ref{sec:related} we outline related work and in Section \ref{sec:conclusion} we discuss future work.

\section{Preliminaries}
\label{sec:prelims}
We start by reviewing some important definitions.

\subsection{Differential Privacy}
Differential privacy is a notion of privacy which ensures that, on two closely related data sets, the output distribution of a randomized algorithm is similar. In other words, an adversary who receives the output of a differentially private algorithm has trouble distinguishing whether a specific entry is in the data set.

\begin{definition}[$\varepsilon$-Differential Privacy] A randomized algorithm $A: \inputspace{} \rightarrow \outputspace$ is $\varepsilon$-differentially private for size $n$ data sets if, for every pair of neighboring data sets $\mathbf{x}, \mathbf{x}^{\prime}$, for all $E \subseteq \outputspace$,
    \begin{equation}
    \frac{\pr{A(\mathbf{x}) \in E}}{\pr{A\left(\mathbf{x}^{\prime}\right) \in E}} \leq e^{\varepsilon}.
\end{equation}
    \end{definition}
Where neighboring data sets are data sets that differ in the value of exactly one entry.
An important property of $\varepsilon$-Differential Privacy is that for all possible outputs $y$ in $\outputspace$, 
\begin{equation}
\frac{\pr{A(\mathbf{x}) =y}}{\pr{A\left(\mathbf{x}^{\prime}\right) =y}} \leq e^{\varepsilon} 
\implies
\frac{\pr{A(\mathbf{x}) \in E}}{\pr{A\left(\mathbf{x}^{\prime}\right) \in E}} \leq e^{\varepsilon}
\label{eqn:singletons} 
\end{equation}

In other words, for $\varepsilon$-DP, it is sufficient to look at all singleton sets in the event space. In subsequent sections we will directly use the definition of differential privacy which is quantified over the singleton sets, and refer to the quantity $\frac{\pr{A(\mathbf{x}) = y}}{\pr{A\left(\mathbf{x}^{\prime}\right) = y}}$ as the \emph{likelihood ratio}. We will also limit our method to algorithms with discrete, finite inputs and outputs (i.e. \inputspace{} and \outputspace{} are discrete and finite).

\subsection{Accuracy}
We consider a widely-used notion of ($\alpha,\beta$)-accuracy from the differential privacy literature \cite{dwork2014algorithmic}. Intuitively, this measures the probability that the algorithm's output is within an $\alpha$ ball around the ``true'' or ``target'' value for the associated input.

\begin{definition}[($\alpha,\beta$)-Accuracy]
    Given $\alpha \geq 0$ and $\beta \in [0,1]$, a randomized algorithm $A:\inputspace \rightarrow \outputspace$ is ($\alpha,\beta$)-accurate if for all $\mathbf{x}\in \inputspace$, 
    \begin{equation}
        \pr{|A(\mathbf{x}) - \targetvals_\mathbf{x}| \leq \alpha} \geq 1-\beta
    \end{equation}
    where $\targetvals_\mathbf{x}$ correspond to the target outputs for each input.
\end{definition}

\subsection{Weighted Model Counting}
\subsubsection{Model Counting}
Let $\varphi$ be a Boolean formula over a set of variables. The model count of $\varphi$ is the number of solutions to $\varphi$, written as $|\{m\models\varphi\}|$ where $\{m\models\varphi\}$ is described as ``the set of models $m$ that entail $\varphi$''. For example, $\varphi = a\lor b$ where $\varphi$ is defined over $\{a,b\} \implies |\{m\models\varphi\}|=3$ since there are three satisfying assignments to $\varphi$.

\subsubsection{Weighted Boolean Formula (WBF)}
A \emph{weighted Boolean formula} is a tuple $(\varphi,w)$ where $\varphi$ is a Boolean formula over literals in $L$ and $w:L\rightarrow \real$ is its weighting function that maps literals in $\varphi$ to weights, where literals are a set of variables and their negations.

\subsubsection{Weighted Model Counting (WMC)}
We find the \textit{weighted model count} (WMC) of WBF $(\varphi,w)$ by computing
\begin{equation}
    \wmc(\varphi_w)=\sum_{m\models \varphi} \prod_{\ell\in m} w(\ell)
\end{equation}
where $m$ is the set of models, or satisfying assignments of $\varphi$, and $\ell$ represents the set of literals in $m$.

For example, we can find the weighted model count of $\varphi = a\lor b$, where $\varphi$ is defined over variables $a$ and $b$, by defining the weight function $w$ over the literals as $w(a)=1/3$, $w(b)=3/4$, $w(\overline{a})=2/3$, and $w(\overline{b})=1/4$ and computing
\begin{equation}
    \wmc(\varphi_w)=1/3\cdot3/4 + 1/3\cdot 1/4 + 2/3\cdot 3/4.
\end{equation}

\subsubsection{WMC for Probabilistic Algorithms}
We can define a WBF for any finite, discrete probabilistic algorithm $A:\inputspace\rightarrow\outputspace$ by ensuring that the probabilities of each WBF assignment correspond with probabilities of input/output pairs of the algorithm. More formally,
\begin{definition}[WBF of a probabilistic algorithm]
    Given a randomized algorithm $A:
    \inputspace\rightarrow\outputspace$ with discrete, finite $\inputspace,~\outputspace$, we say that $(\varphi,w)$ is a WBF of $A$ if $\forall(\mathbf{x},y)\in\inputspace\times\outputspace$, there exists exactly one assignment $a$ to $(\varphi,w)$ such that $\wmc(\varphi\mid a,w)=\pr{A(\mathbf{x})=y}$
\end{definition}

We use a shorthand $\varphi\mid (\mathbf{x},y)$ to indicate the WBF instantiated with the assignment $a$ for which $\wmc(\varphi\mid a,w)=\pr{A(\mathbf{x})=y}$.

For example, we consider binary randomized response for $n=2$ as described in Fig. \ref{fig:diagram}. In the randomized response algorithm, clients report a bit message (represented by $x_1$ and $x_2$ here) with probability $1-\lambda$, and flip this bit with probability $\lambda$ for some coin flip parameter $\lambda$. More formally, $\rr^{2}(x_1,x_2)=(y_1,y_2)$ such that $y_1= x_1~w.p.~(1-\lambda)$ else $1-x_1$ and $y_2= x_2~w.p.~(1-\lambda)$ else $1-x_2$. In this case, the WBF for randomized response is the tuple $(\varphi_{\rr^2},w_{\rr^2})$ where $\varphi_{\rr^2}=y_1 \leftrightarrow ((\overline{\theta}_1\land x_1)\lor(\theta_1\land \overline{x}_1))\land y_2 \leftrightarrow ((\overline{\theta}_2\land x_2)\lor(\theta_2\land \overline{x}_2))$ defined over $x_1,x_2,y_1,y_2,\theta_1,\theta_2$. We define weight function $w(\theta_i)=1-\lambda$, or the probability that the value of $x_i$ does not flip, and $w(\overline{\theta}_i)=\lambda$ as the probability that $x_i$ does flip for $i\in\{1,2\}$. For all other literals, we set $w(\cdot)=1$.

We can see that for all $(x_1,x_2,y_1,y_2)$, $\wmc(\varphi_{\rr^2}\mid (x_1,x_2,y_1,y_2),w_{\rr^2})=\pr{\rr^2(x_1,x_2)=(y_1,y_2)}$. For example, for $(x_1,x_2,y_1,y_2)=(0,0,1,0)$, $\wmc(\varphi_{\rr^2}\mid (0,0,1,0),w_{\rr^2})=w_{RR^2}(\theta_1)\cdot w_{RR^2}(\theta_2)=\lambda\cdot(1-\lambda)=\pr{\rr^2(0,0)=(1,0)}$.

\section{Weighted Model Counting for Synthesizing Tight Privacy and Accuracy Bounds}
\label{sec:wmc_for_synthesis}
We start by introducing the two main synthesis problems we will address in this paper. Then, we describe how weighted model counting can be used to efficiently perform probabilistic inference in an exhaustive solution. In Section \ref{sec:symmetries}, we will outline a framework for leveraging symmetries in the DP algorithms to further improve on this solution.

\subsection{Synthesis Problems}
The first problem is the \emph{privacy bound synthesis problem}. We want to find the triple of neighboring input and output values which minimize the likelihood ratio.

\begin{problem}[Privacy Bound Synthesis]
    Given a randomized algorithm $A: \inputspace \rightarrow \outputspace$ with discrete, finite $\inputspace,~\outputspace$, find the triple $(\mathbf{x},\mathbf{x}',y)$ that maximizes $\frac{\pr{A(\mathbf{x}) = y}}{\pr{A(\mathbf{x}') = y}}$.
\end{problem}

If we take solution $(\mathbf{x},\mathbf{x}',y)$ to the privacy bound synthesis problem to find $e^\varepsilon=\frac{\pr{A(\mathbf{x}) = y}}{\pr{A(\mathbf{x}') = y}}$, we have a tight $\varepsilon$-DP bound for algorithm $A$. 

The second synthesis problem is the \emph{tight accuracy bound synthesis problem}. Here we want to find the input to $A$ which minimizes the accuracy probability with respect to a given $\alpha$ bound.

\begin{problem}[Accuracy Bound Synthesis]
    Given a randomized algorithm $A: \inputspace \rightarrow \outputspace$ with discrete, finite $\inputspace,~\outputspace$, and set of target values $\targetvals$, and 
    $\alpha \geq 0$, find $\mathbf{x}$ which minimizes $\pr{|A(\mathbf{x}) - \targetvals_\mathbf{x}| \leq \alpha}$.
\end{problem}

If we take solution $\mathbf{x}$ from the Accuracy Bound Synthesis problem to find $1-\beta=\pr{|A(\mathbf{x}) - \targetvals_\mathbf{x}| \leq \alpha}$, we have a tight ($\alpha,\beta$)-accuracy bound for algorithm $A$.

\subsection{WMC and Probabilistic Programming For Privacy Bound Synthesis}
We start by looking at the most straightforward solution to the privacy bound synthesis problem, which is to iterate through every input/neighbor/output triple and compute the likelihood ratio for that set. In other words, we consider a solution where, for $A:\inputspace\rightarrow\outputspace$ we exhaustively compute $\frac{\pr{A(\mathbf{x}) = y}}{\pr{A(\mathbf{x}') = y}}$ for each $(\mathbf{x},\mathbf{x}',y)\in  \inputspace\times\inputspace\times\outputspace$ where $\mathbf{x},\mathbf{x}'$ are neighbors. We will refine this solution further in Section \ref{sec:symmetries}, but for now we use this simple solution to illustrate the utility of WMC for probabilistic inference.

If we consider even this simple exhaustive solution, it is immediately apparent that for a complex randomized algorithm, $A$, it is both necessary and non-trivial to perform probabilistic inference (i.e. compute the probabilities $\pr{A(\mathbf{x}) = y}$) for each input and output pair.

The power of our technique comes from performing probabilistic inference (computing the probability distributions of the algorithm) by compiling a representation of the probabilistic algorithm into a \emph{tractable circuit} for which computing the weighted model count can be performed efficiently~\cite{chavira2008probabilistic,holtzen2020generating}. 

We consider our target data structure to be a \emph{binary decision diagram} (BDD). A BDD is a tree-like data structure that represents boolean formulas. We prefer to use this data structure for WMC computation because BDDs can leverage shared structure in the encoding to simplify a complicated function with many free variables into a compact, efficiently queryable circuit. As we see in Fig. \ref{fig:diagram}, a program with two inputs, two outputs, and two weighted coin flip parameters can be encoded into a BDD with only two internal nodes when the input is fixed. When we compare this to other popular methods for exact probabilistic inference such as Markov chains, we see in Section \ref{sec:eval} that a BDD is able to scale much more efficiently with the state size on static systems such as DP algorithms. 

For example, prior work has used Markov chain model checking for exact DP verification \cite{liu2018model,liu2022verifying}. This technique requires an explicit Markov chain model of the probability distribution, which is optimized for dynamic models with small state space which model long-running services. In this case, computation becomes inefficient quickly as state size increases, as is the case in DP algorithms \cite{HoltzenCAV21}.

Using BDDs for WMC computation is therefore more efficient for our problem space than in prior attempts, but we also desire an expressive encoding method for the randomized algorithm. To use WMC to perform inference on the randomized algorithm, we must encode the algorithm as a weighted Boolean formula with free variables that represent the inputs, outputs, and coin flips/distributions, and craft a weighting function to represent the randomness of the algorithm. For some algorithms, this is straightforward and little manual work is required to discover the right formula and weighting function. We provide an example of this manual process for randomized response in Section \ref{sec:rr}. However, for many algorithms, this manual process can be difficult, and the resulting weighted Boolean formula can be unnecessarily large or otherwise inefficient to work with.

In recent years, there has been a significant amount of work on developing probabilistic programming languages that can efficiently compile weighted Boolean formulas from probabilistic programs \cite{holtzen2020scaling}. Importantly, these weighted Boolean formulas are optimized to efficiently compile into tractable data structures such as BDDs that can be efficiently queried to find the weighted model count. This means that we can write our randomized algorithm as a probabilistic program in these languages and to easily compile it into a WBF and perform efficient WMC. We show an example of this full pipeline in Fig. \ref{fig:diagram}.

In summary, using WMC for performing probabilistic inference is a method of efficient computation of the probabilities needed to find tight bounds. Furthermore, recently developed probabilistic programming languages provide approachable and intuitive methods for encoding DP algorithms. This is a big improvement over prior work which required that the DP algorithm be explicitly encoded as a Markov chain~\cite{liu2018model,liu2022verifying}.

\subsection{WMC for Exhaustive Accuracy Bound Synthesis}
Similarly to our exhaustive approach for solving the privacy bound synthesis problem, we can also leverage the benefits of WMC in an exhaustive approach for solving the accuracy bound synthesis problem. Because the target value of the DP algorithm (i.e. the value of the query with no added noise) is non-probabilistic, we can formulate an optimization which which will directly use probabilistic inference to solve the accuracy bound synthesis problem.

\begin{thm}
    \label{thm:acc_minimizer}
    Given randomized $A:\inputspace\rightarrow\outputspace$, with discrete, finite $\inputspace$ and $\outputspace$, set of target values $\targetvals$, and $\alpha\geq 0$, the $\mathbf{x}$ which minimizes $\sum_{y\in[\targetvals_\mathbf{x}-\alpha,\targetvals_\mathbf{x}+\alpha]}P(A(\mathbf{x})=y)$ also minimizes $Pr(|A(\mathbf{x}) - \targetvals_\mathbf{x}| \leq \alpha)$
\end{thm}
\begin{proof}
    \begin{align*}
        &Pr(|A(\mathbf{x}) - \targetvals_\mathbf{x}| \leq \alpha) \\
        &= Pr(\targetvals_\mathbf{x} - \alpha \leq A(\mathbf{x})\leq \targetvals_\mathbf{x} + \alpha)\\
        &=\sum_{y\in[\targetvals_\mathbf{x}-\alpha,\targetvals_\mathbf{x}+\alpha]}\pr{A(\mathbf{x})=y}.
    \end{align*}
\end{proof}

We can therefore also use weighted model counting for the simple exhaustive approach for $A:\inputspace\rightarrow\outputspace$ where we compute $\sum_{y\in[\targetvals_\mathbf{x}-\alpha,\targetvals_\mathbf{x}+\alpha]}\pr{A(\mathbf{x})=y}$ for each $\mathbf{x}\in \inputspace$. Again, in Section \ref{sec:symmetries}, we will show a further refinement of this exhaustive algorithm.

\section{Framework for Leveraging Symmetries}
\label{sec:symmetries}
Despite the gains in efficiency and usability we get from using WMC for probabilistic inference, the exhaustive solution still results in state space explosion as the size of the input and output space increase. Fortunately, many differentially private algorithms have inherent symmetry. We leverage these symmetries to make our WMC solutions more tractable by introducing algorithms for computing the privacy and accuracy bounds which use \textit{inference}, \textit{privacy}, and \textit{accuracy sets} for limiting the search space.

\subsection{Inference Algorithm}
Our refined solutions for privacy and accuracy bound synthesis both use probabilistic inference via weighted model counting as a subprocess. We provide an algorithm for pre-computation of necessary probabilities to help with modularization and simplification of the algorithm and analysis.

In Algorithm \ref{alg:inf}, we describe how we compute a matrix of assignment probabilities for a given restricted \textit{inference set}. An \emph{inference set} $\infset_{A}\subseteq\inputspace\times\outputspace$ for $A:\inputspace\rightarrow\outputspace$ is a set of input/output pairs on which to do probabilistic inference.

\begin{algorithm}
    \caption{$\mathsf{INFERENCE}$}
    \label{alg:inf}
    \begin{algorithmic}[1]
    \renewcommand{\algorithmicrequire}{\textbf{Input:}}
    \renewcommand{\algorithmicensure}{\textbf{Output:}}
    \REQUIRE WBF $(\varphi,w)$ of $A:\inputspace\rightarrow\outputspace$ and inf. set $\infset$
    \ENSURE \probmat{} a matrix of probabilities 
    \STATE Initialize $\probmat{}$ matrix
    \FOR{$(\mathbf{x},y)$ in \infset }
        \STATE $\probmat(\mathbf{x},y)\gets \wmc(\varphi\mid (\mathbf{x},y))$
    \ENDFOR
    \RETURN \probmat{}
    \end{algorithmic}
\end{algorithm}

 Algorithm \ref{alg:inf} runs in $O(\mathsf{WMC}\cdot|\infset|)$ where $\mathsf{WMC}$ is the complexity of the WMC computation. We will show in Section \ref{sec:rr} that for some algorithms, we can reduce the size of the inference set such that the complexity is linear in the length of the input vector, multiplied by the complexity of the WMC computation.

 Returning to our example of binary randomized response, an inference set for $\rr^2$ could be $\infset_{\rr^2}=\{([0,0],[0,0]),~([1,0],[0,0]),~([1,1],[0,0])\}$. Table \ref{tab:inf} shows the output of this computation. We show in Section \ref{sec:rr} that this table grows linearly in $n$.
 
\begin{table}[]
    \centering
    \resizebox{.4\linewidth}{!}{%
    \begin{tabular}{lllll}
        &                                & \multicolumn{3}{c}{input}                                                                          \\
        & \multicolumn{1}{l|}{}          & \multicolumn{1}{l|}{{[}0,0{]}}       & \multicolumn{1}{l|}{{[}1,0{]}}                 & {[}1,1{]}   \\ \cline{2-5} 
  output & \multicolumn{1}{l|}{{[}0,0{]}} & \multicolumn{1}{l|}{$(1-\lambda)^2$} & \multicolumn{1}{l|}{$\lambda\cdot(1-\lambda)$} & $\lambda^2$
  \end{tabular}%
  }
\caption{Output of Algorithm \ref{alg:inf} for $A=\rr^2$, $\infset_{\rr^2}=\{([0,0],[0,0]),~([1,0],[0,0]),~([1,1],[0,0])\}$}
\label{tab:inf}
\end{table}

% PRIVACY
\subsection{Finding Tight Privacy Bound}
Once we have computed the necessary probabilities using weighted model counting, we can use these probabilities to find the maximum likelihood ratio. Because we consider only neighboring inputs, and due to the symmetries in differential privacy algorithms, we do not have to enumerate all pairs of inputs and all output events to cover all sufficient ratios. We therefore define our algorithm over a set of inputs and outputs which are a sufficient subset of all computable likelihood ratios. 

\begin{definition}[Privacy Set]
    \label{def:compset}
    A \emph{privacy set} \compset{} for $A:\inputspace\rightarrow\outputspace$ is a set of tuples of the form $(\mathbf{x}_\compset,\mathbf{x}_\compset',y_\compset)$ such that
    \begin{enumerate}
        \item $(\mathbf{x}_\compset,\mathbf{x}_\compset',y_\compset)\in \inputspace\times\inputspace\times\outputspace$
        \item $\forall (\mathbf{x},\mathbf{x}',y)\in  \inputspace\times\inputspace\times\outputspace$ where $\mathbf{x}, \mathbf{x}'$ are neighbors, there exists $(\mathbf{x}_\compset,\mathbf{x}'_\compset,y_\compset)\in\compset$ such that $$\frac{Pr(A(\mathbf{x}) = y)}{Pr(A(\mathbf{x}') = y)}=\frac{Pr(A(\mathbf{x}_\compset) = y_\compset)}{Pr(A(\mathbf{x}'_\compset) = y_\compset)}.$$
    \end{enumerate}
\end{definition}

We present the algorithm for computing worst case privacy bounds in Algorithm \ref{alg:priv}. We use the $\mathsf{INFERENCE}$ algorithm (Algorithm \ref{alg:inf}) as a subprocess to compute exact probabilistic inference via weighted model counting.

\begin{algorithm}
    \caption{Privacy Bound Synthesis}
    \label{alg:priv}
    \begin{algorithmic}[1]
    \renewcommand{\algorithmicrequire}{\textbf{Input:}}
    \renewcommand{\algorithmicensure}{\textbf{Output:}}
    \REQUIRE WBF $(\varphi,w)$ of $A:\inputspace\rightarrow\outputspace$, privacy set $\compset$, and inference set $\infset$ such that $\forall (\mathbf{x}_\compset,\mathbf{x}_\compset',y_\compset)\in \compset$, $(\mathbf{x}_\compset, y_\compset)\in \infset$ and $(\mathbf{x}_\compset', y_\compset)\in \infset$
    \ENSURE Maximum likelihood ratio $p$ and worst case assignments $c=(\mathbf{x},\mathbf{x}',y)$
    \STATE $\probmat \gets \mathsf{INFERENCE}(\varphi, \infset)$
    \STATE $p\gets 0$, $c\gets\varnothing$
    \FOR{$(\mathbf{x},\mathbf{x}',y)$ in \compset}
        \IF{$\frac{M(\mathbf{x},y)}{M(\mathbf{x}',y)}>p$}
        \STATE $p\gets \frac{M(\mathbf{x},y)}{M(\mathbf{x}',y)}$
        \STATE $c\gets (\mathbf{x},\mathbf{x}',y)$
        \ENDIF
    \ENDFOR
    \RETURN $p,c$
    \end{algorithmic}
\end{algorithm}

Returning to our binary randomized response example, we can use $\infset_{\rr^2}=\{([0,0],[0,0]),$ $([1,0],[0,0]),$ $([1,1],[0,0])\}$, and $\compset_{\rr^2}=\{([0,0],[1,0],[0,0]),$ $([1,0],[0,0],[0,0]),$ $([1,0],[1,1],[0,0]),$ $([1,1],[1,0],[0,0])\}$ and set $\lambda=0.2$. In this case, the output of Algorithm \ref{alg:priv} would be $c=([0,0],[1,0],[0,0])$, $p=4$. This means that $([0,0],[1,0],[0,0])$ is the worst case assignment for $\rr^2$ and $e^\varepsilon=4$ is a tight privacy bound.

\subsubsection{Correctness and Complexity of Privacy Bound Synthesis Algorithm}
To verify correctness of our algorithm, we must show that the maximal likelihood ratio found using the potentially restricted inference and privacy set is the same as the maximal likelihood ratio using exhaustive inference and privacy sets (i.e. $\infset=\inputspace\times\outputspace$ and $\compset=\inputspace\times\inputspace\times\outputspace$ where all $x,x'$ are neighbors). 

\begin{thm}
    The output $c$ of Algorithm \ref{alg:priv} is a solution to the Privacy Bound Synthesis problem.
\end{thm}
\begin{proof}
    Let WBF $(\varphi,w)$ be the WBF of algorithm $A:\inputspace\rightarrow\outputspace$ with inference and privacy sets $\infset,~\compset$. From the definition of Algorithm \ref{alg:inf}, $\forall(\mathbf{x},y)\in\infset$, $M(\mathbf{x},y)=\pr{A(\mathbf{x},y)}$. Therefore, by the definition of the Privacy Bound Synthesis algorithm, output $c$ maximizes $\frac{M(\mathbf{x}_\compset,y_\compset)}{M(\mathbf{x}'_\compset,y_\compset)}$ over all $(\mathbf{x},\mathbf{x}',y)\in\compset$. By the definition of a privacy set, for all $(\mathbf{x},\mathbf{x}',y)\in \inputspace\times\inputspace\times\outputspace$, where $\mathbf{x},~\mathbf{x}'$ are neighbors, there exists a $(\mathbf{x},\mathbf{x}',y)\in\compset$ such that $\frac{\pr{A(\mathbf{x}) = y}}{\pr{A(\mathbf{x}') = y}}=\frac{\pr{A(\mathbf{x}_\compset) = y_\compset}}{\pr{A(\mathbf{x}'_\compset) = y_\compset}}$. Therefore, $c$ maximizes $\frac{\pr{A(\mathbf{x}) = y}}{\pr{A(\mathbf{x}') = y}}$ over all $(\mathbf{x},\mathbf{x}',y)\in \inputspace\times\inputspace\times\outputspace$ where $\mathbf{x},~\mathbf{x}'$ are neighbors and is therefore a solution to the Privacy Bound Synthesis problem.
\end{proof}

Algorithm \ref{alg:priv} runs in $O(\mathsf{WMC}\cdot|\infset|+|\compset|)$. Again, in Section \ref{sec:rr}, we will show a case where the privacy set can be restricted such that $|\infset|$ and $|\compset|$ are linear in $n$.

When we concretely define the inference set $\infset$ and privacy set $\compset$ for a specific algorithm, we must prove two things:
\begin{enumerate}
    \item \compset{} satisfies the definition of a valid privacy set in Def. \ref{def:compset}, and
    \item All the neighboring assignments in \compset{} are present in \infset{}, i.e. $\forall (\mathbf{x}_\compset,\mathbf{x}_\compset',y_\compset)\in \compset$, $(\mathbf{x}_\compset, y_\compset)\in \infset$ and $(\mathbf{x}_\compset', y_\compset)\in \infset$.
\end{enumerate}

We provide an example of one such instantiation for randomized response in Section \ref{sec:rr}.

% ACCURACY
\subsection{Finding Tight Accuracy Bound}
In the case of accuracy bound synthesis, we can also limit the number of values we need to compute to leverage symmetries in the problem to reduce the runtime of the accuracy computation. We introduce the accuracy set as follows.
\begin{definition}[Accuracy Set]
    \label{def:accset}
    An \emph{accuracy set} \accset{} for $A:\inputspace\rightarrow\outputspace$ with corresponding vector of target outputs $\targetvals$ is a set $\accset$ such that
    \begin{enumerate}
        \item $\accset\subseteq\inputspace$
        \item $\forall \mathbf{x}\in  \inputspace$, there exists $\mathbf{x}_\accset\in\accset$ such that 
        $$\smashoperator{\sum_{y\in[\targetvals[\mathbf{x}]-\alpha,\targetvals[\mathbf{x}]+\alpha]}}P(A(\mathbf{x})=y)=$$ $$\smashoperator{\sum_{y\in[\targetvals[\mathbf{x}_\accset]-\alpha,\targetvals[\mathbf{x}_\accset]+\alpha]}}P(A(\mathbf{x}_\accset)=y).$$
    \end{enumerate}
\end{definition} 

We present our solution in Algorithm \ref{alg:acc}. Again, use the $\mathsf{INFERENCE}$ algorithm (Algorithm \ref{alg:inf}) as a subprocess to compute exact probabilistic inference via WMC.

\begin{algorithm}
    \caption{Accuracy Bound Synthesis}
    \label{alg:acc}
    \begin{algorithmic}[1]
    \renewcommand{\algorithmicrequire}{\textbf{Input:}}
    \renewcommand{\algorithmicensure}{\textbf{Output:}}
    \REQUIRE WBF $(\varphi,w)$ of $A:\inputspace\rightarrow\outputspace$, accuracy set \accset, vector of target outputs \targetvals{}, accuracy parameter $\alpha$, and inference set $\infset$ that contains all $(\mathbf{x},y)$ pairs such that $x\in\accset$ and $y\in[\mathcal{V}_\mathbf{x}-\alpha,\mathcal{V}_\mathbf{x}+\alpha]$
    \ENSURE Minimal probability $p$ and worst case input $\mathbf{acc}$
    \STATE $\probmat \gets \mathsf{INFERENCE}(\varphi, \infset)$
    \STATE $p\gets \infty$, $\mathbf{acc}\gets\varnothing$
    \FOR{$\mathbf{x}$ in ($\infset_{\inputspace}$)}
        \IF{$\sum_{y\in[\targetvals_\mathbf{x}-\alpha,\targetvals_\mathbf{x} + \alpha]}\probmat[\mathbf{x},y]<p$}
        \STATE $p \gets\sum_{y\in[\targetvals_\mathbf{x}-\alpha,\targetvals_\mathbf{x} + \alpha]}\probmat[\mathbf{x},y]$
        \STATE $\mathbf{acc}\gets \mathbf{x}$
        \ENDIF
    \ENDFOR
    \RETURN $p, \mathbf{acc}$
    \end{algorithmic}
\end{algorithm}

Again, we return to our binary randomized response example. Since we are considering accuracy, we have to wrap the raw output to compute some quantifiable query. For example, we can add a counting query wrapper, such that the output of $\rrcount^2(x_1,x_2)=sum(\rr^2(x_1,x_2))$. Here we set $\alpha=1,~\lambda=0.2,~\infset_{\rrcount^2}=
\{([0,0],0),$ $([1,1],0)
,$ $([0,0],1),$ $([1,0],1),$ $([1,1],1)
,$ $([1,0],2),$ $([1,1],2)
\}$, and $\accset_{\rrcount^2}=\{[0,0],[1,0],[1,1]\}$ where $\targetvals_{[0,0]}=0,\targetvals_{[1,0]}=1,\targetvals_{[1,1]}=2$.
In this case, the output of Algorithm \ref{alg:acc} would be $\mathbf{acc}=[0,1]$, $p=.96$. This means that $[1,0]$ is the worst case accuracy assignment for $\rrcount^2$ and $1-\beta=.96$ is a tight accuracy bound.

\subsubsection{Correctness and Complexity of Accuracy Bound Synthesis}
To verify correctness of our accuracy algorithm, we must show that the minimal accuracy probability found using the potentially restricted inference and accuracy set is the same as the minimal accuracy probability using exhaustive inference and accuracy sets.
\begin{thm}
    The output $\mathbf{acc}$ of Algorithm \ref{alg:acc} is a solution to the Accuracy Bound Synthesis problem.
\end{thm}
\begin{proof}
    Let WBF $(\varphi,w)$ be the WBF of $A:\inputspace\rightarrow\outputspace$, $\targetvals$ be a vector of target outputs of $A$, $\alpha\geq 0$ and let $\infset$, $\accset$ be inference and accuracy sets of $A$ such that $\infset$ contains all $(\mathbf{x}_{\accset},y)$ where $\mathbf{x}_{\accset}\in\accset$ and $y\in[\mathcal{V}[\mathbf{x}_{\accset}]-\alpha,\mathcal{V}[\mathbf{x}_{\accset}]+\alpha]$. We know that $\mathbf{acc}$ minimizes $\sum_{y\in[\targetvals[\mathbf{x}_{\accset}]-\alpha,\targetvals[\mathbf{x}_{\accset}] + \alpha]}\probmat[\mathbf{x}_{\accset},y]$ by the definition of Algorithm \ref{alg:acc}, and that $\probmat[\mathbf{x}_{\accset},y]=\pr{A(\mathbf{x}_{\accset}) = y}$ by the definition of Algorithm \ref{alg:inf}. Therefore, $\mathbf{acc}$ minimizes $\pr{A(\mathbf{x}_{\accset}) = y}$ over $\mathbf{x}_{\accset}\in\accset$. By the definition of an accuracy set, for all $\mathbf{x}\in \inputspace$, there exists a $\mathbf{x}\in\accset$ such that ${\sum_{y\in[\targetvals[\mathbf{x}]-\alpha,\targetvals[\mathbf{x}]+\alpha]}}P(A(\mathbf{x})=y)={\sum_{y\in[\targetvals[\mathbf{x}_\accset]-\alpha,\targetvals[\mathbf{x}_\accset]+\alpha]}}P(A(\mathbf{x}_\accset)=y)$. Therefore, $c$ maximizes $\frac{\pr{A(\mathbf{x}) = y}}{\pr{A(\mathbf{x}') = y}}$ over all $(\mathbf{x},\mathbf{x}',y)\in \inputspace\times\inputspace\times\outputspace$. By Theorem \ref{thm:acc_minimizer}, this means that $\mathbf{acc}$ minimizes $Pr(|A(\mathbf{x}) - \targetvals_\mathbf{x}| \leq \alpha)$ over all $\mathbf{x}\in\inputspace$ and is therefore a solution to the Accuracy Bound Synthesis problem.
\end{proof}

Algorithm \ref{alg:acc} runs in $O(\mathsf{WMC}\cdot|\infset|+|\accset|\cdot\alpha)$.

Like with the privacy algorithm, there are two main proof requirements when we design a concrete inference set $\infset$ and accuracy set $\accset$ to use in the accuracy bound synthesis algorithm. These requirements are:
\begin{enumerate}
    \item $\accset$ satisfies the definition of a valid accuracy set in Definition \ref{def:accset} and
    \item all necessary input/output assignments are present in $\infset$, i.e. $\forall x\in\accset$ and $y\in[\mathcal{V}[\mathbf{x}]-\alpha,\mathcal{V}[\mathbf{x}]+\alpha]$,  $(\mathbf{x},y)\in\infset{}$.
\end{enumerate}

\section{Randomized Response Case Study}
\label{sec:rr}
We have been using an example of binary randomized response for $n=2$. In this section, we generalize this solution to any $n$ and provide a detailed explanation of how to utilize these tools for a tractable solution to the tight privacy and accuracy bound problem.

\begin{algorithm}
    \caption{Randomized Response (\rr)}
    \label{alg:rr}
    \begin{algorithmic}[1]
    \renewcommand{\algorithmicrequire}{\textbf{Input:}}
    \renewcommand{\algorithmicensure}{\textbf{Output:}}
    \REQUIRE Bit array $\mathbf{x}$ of true client messages. 
    \ENSURE Bit array $\mathbf{y}$ of randomized client messages.
    \STATE Initialize bit array $\mathbf{y}$ of length $|\mathbf{x}|$
    \FOR{$i\in\{1,\dots,|\mathbf{x}|\}$}
        \STATE $\mathbf{y}[i]\gets 1-\mathbf{x}[i]$ with probability $\lambda$, otherwise $\mathbf{x}[i]$.
    \ENDFOR
    \RETURN $\mathbf{y}$
    \end{algorithmic}
\end{algorithm}

In this section, we provide a WBF for \rr{} and show the complexity of the exhaustive solution. We then analyze the correctness of a solution to the privacy bound synthesis problem that leverages the inherent symmetries in \rr{} that runs linearly in $n$ (multiplied by $\wmctime$). We also solve the accuracy bound synthesis problem in time quadratic in $n$ (again multiplied by $\wmctime$).

%  WBF
\subsection{Weighted Boolean Formula for RR}
We can manually craft a WBF $(\varphi,w)$ for randomized response where $n$ is the number of clients. We set 
\begin{align}
    \varphi=\bigwedge\limits_{i=1}^n y_i \leftrightarrow ((\overline{\theta}_i\land x_i)\lor(\theta_i\land \overline{x}_i))    
\end{align}
where each $x_i$ corresponds to the bit value in the input vector of \rr{}, and the $y_i$'s likewise correspond with the output vectors. The $\theta$ values correspond with the coin flips in the randomized portion of the algorithm.

We set the weighting function to be $w(\theta_i)=1-\lambda$, or the probability that the value of $x_i$ does not flip, and $w(\overline{\theta}_i)=\lambda$ as the probability that $x_i$ does flip. For all other literals, we set $w(\cdot)=1$.

\begin{thm}
    $(\varphi,w)$ is a valid WBF of \rr.
\end{thm}
\begin{proof}
    Let $(\mathbf{x},\mathbf{y})\in\inputspace\times\outputspacerr$. For each $i\in\{1,\dots,n\}$ assign the $x_i$ and $y_i$ literals in $\varphi$ the values of $\mathbf{x}[i]$ and $\mathbf{y}[i]$, and $\overline{x}_i=1-\mathbf{x}[i]$. If $\mathbf{x}[i]=\mathbf{y}[i]$, then for $y_i \leftrightarrow ((\overline{\theta}_i\land x_i)\lor(\theta_i\land \overline{x}_i))$ to be satisfied, $\theta_i=0$. If $\mathbf{x}[i]\not=\mathbf{y}[i]$, then for $y_i \leftrightarrow ((\overline{\theta}_i\land x_i)\lor(\theta_i\land \overline{x}_i))$ to be satisfied, $\theta_i=1$. 
    
    This is the only satisfying assignment of $\varphi$ for $x,~y$ assigned as described, therefore there is only one model, and the weighted model count is therefore $\prod_{\ell\in m} w(\ell)$. 
    
    The product of weights of literals is $w(x_i)\cdot w(y_i)\cdot w(\overline{\theta}_i)=1\cdot1\cdot (1-\lambda)$ for $i$ s.t. $\mathbf{x}[i]=\mathbf{y}[i]$ and $w(x_j)\cdot w(y_j)\cdot w(\theta_j)=1\cdot1\cdot \lambda$ for $i$ s.t. $\mathbf{x}[i]\not=\mathbf{y}[i]$, so $\wmc((\varphi,w))=(1-\lambda)^{\# i \text{ s.t. } \mathbf{x}[i]=\mathbf{y}[i]}\lambda^{\# i \text{ s.t. } \mathbf{x}[i]\not=\mathbf{y}[i]}=\pr{\rr(\mathbf{x})=\mathbf{y}}.$
\end{proof}

We provide this analysis to demonstrate what a valid WBF for \rr{} would look, however, as discussed in Section \ref{sec:wmc_for_synthesis}, this is not necessarily the best WBF for computing the WMC of $\rr$. In Section \ref{sec:eval} we use a probabilistic programming language to find the WBF, compile to a compact BDD, and efficiently query the WMC.

% EXHAUSTIVE
\subsection{Exhaustive Solution}
It is evident that for all algorithms, setting $\infset_{\rr}=\inputspace\times\outputspace$, $\compset=\inputspace\times\inputspace\times\outputspace$, and $\accset_{\rr}=\inputspace$ gives the correct bounds for both privacy and accuracy because we are optimizing the bounds over all possible inputs and outputs. However, because $|\inputspace|=2^n$ and $|\outputspace|=2^n$ for randomized response, when we analyze this solution, we see that the runtime of the inference algorithm is $O(\wmctime{} \cdot|\infset_{\rr}|)=O(\wmctime{} \cdot 4^n)$, the runtime of the privacy algorithm is $O(\wmctime{} \cdot|\infset_{\rr}|\cdot|\compset|)=O(\wmctime{}\cdot 4^n+16^n)$ and the runtime of the accuracy algorithm is $O(\wmctime{}\cdot|\infset_{\rr}| + |\accset_{\rr}|\cdot\alpha)=O(\wmctime{} \cdot 4^n + 2^n\cdot\alpha)$. We therefore need to find a way to improve this runtime for randomized response.

% PRIVACY
\subsection{Leveraging Symmetries for Finding Privacy Bound}
Because of the independence properties between different clients responses, there are many inherent symmetries in this algorithm. In this section, we identify an inference and privacy set which satisfy the coverage properties from the previous section, and reduce the complexity by orders of magnitude.

\subsubsection{Counting Flips}
We identify one significant symmetry in \rr{} which is the number of clients whose bit flips. In other words, we can find a representative input/output pair for each number of bit flips that occurs.

\begin{lemma}
    \label{lem:flips}
    Given $i\in\{1,\dots,n\}$, $\forall \mathbf{x}\in\inputspace$ and $\mathbf{y}\in\outputspacerr$ such that $count(\mathbf{x}\lxor\mathbf{y})=i$, $\pr{A(\mathbf{x})=\mathbf{y}}=\pr{A(1^i0^{n-i})=0^n}$.
\end{lemma}
\begin{proof}
    Let $\mathbf{x}\in\inputspace$ and $\mathbf{y}\in\outputspacerr$ such that $count(\mathbf{x}\lxor\mathbf{y})=i$. This means that there are $i$ entries such that $\mathbf{x}[i]\not=\mathbf{y}[i]$. Therefore, 
    $\pr{\rr(\mathbf{x})=\mathbf{y}}=(1-\lambda)^{\# i \text{ s.t. } \mathbf{x}[i] =\mathbf{y}[i]}\lambda^{\# i \text{ s.t. } \mathbf{x}[i] \not=\mathbf{y}[i]}=\pr{A(1^i0^{n-i})=0^n}.$
\end{proof}
We also identify a key fact about the relationship between the number of bit flips in neighboring inputs, specifically that a neighboring input has either one more or one less bit flip with respect to the output. 
\begin{lemma}
    \label{lem:rrneigh}
    For neighboring inputs $\mathbf{x},\mathbf{x}'\in\inputspace$, $count(\mathbf{x}'\lxor y)=count(\mathbf{x}\lxor\mathbf{y})+1$ or $count(\mathbf{x}\lxor \mathbf{y})-1$.
 \end{lemma}
 Lemma \ref{lem:rrneigh} follows directly from the definition of neighboring inputs.

\subsubsection{Constructing Privacy Set}
Because we have these inherent symmetries in the number of flips between each input/output pair, we can define the privacy set to be 
\begin{align}
    \compset_{\rr}=&\{(1^i0^{n-i},1^{i+1}0^{n-(i+1)},0^n)\}_{i\in\{0,\dots,n-1\}} \cup \\ &\{(1^{i}0^{n-i},1^{i-1}0^{n-(i-1)},0^n)\}_{i\in\{1,\dots,n\}}.
\end{align}

We note that $|\compset_{\rr}|=2n$.

\begin{thm}
    $\compset_{\rr}$ is a valid privacy set for $\rr$.
\end{thm}
\begin{proof}
    Let $(\mathbf{x},\mathbf{x}',\mathbf{y})\in\inputspace\times\inputspace\times\outputspacerr$ such that $\mathbf{x},\mathbf{x}'$ are neighbors. 
    
    If $count(\mathbf{x}'\lxor\mathbf{y})=count(\mathbf{x}\lxor\mathbf{y})+1$, then by Lemma \ref{lem:flips} with $count(\mathbf{x}\lxor\mathbf{y})=i$ and $count(\mathbf{x}'\lxor\mathbf{y})=i+1$, $\likelyratio{}=\frac{\pr{A(1^i0^{n-i})=0^n}}{\pr{A(1^{i+1}0^{n-(i+1)})=0^n}}$. By definition, $(1^i0^{n-i},1^{i+1}0^{n-(i+1)},0^n)$ is in $\compset_{\rr}$.

    If $count(\mathbf{x}'\lxor\mathbf{y})=count(\mathbf{x}\lxor\mathbf{y})-1$, then by Lemma \ref{lem:flips} with $count(\mathbf{x}\lxor\mathbf{y})=i$ and $count(\mathbf{x}'\lxor\mathbf{y})=i-1$, $\likelyratio{}=\frac{\pr{A(1^i0^{n-i})=0^n}}{\pr{A(1^{i-1}0^{n-(i-1)})=0^n}}$. By definition, $(1^i0^{n-i},1^{i-1}0^{n-(i-1)},0^n)$ is in $\compset_{\rr}$.

    By Lemma \ref{lem:rrneigh}, this covers all possible cases, and by definition of $\compset_{\rr}$, $\forall(\mathbf{x},\mathbf{x}',\mathbf{y})\in \compset_{\rr}$, $(\mathbf{x},\mathbf{x}',\mathbf{y})\in \inputspace\times\inputspace\times\outputspace$. Therefore $\compset_{\rr}$ is a valid privacy set.
\end{proof}

We can now build a smaller inference set that computes all necessary probabilities that are used in the privacy set. In this case we define 
\begin{equation}
    \infset_{\rr}=\{(1^i0^{n-i},0^n)\}_{i\in\{0,\dots,n\}}.
\end{equation}

\begin{thm}
    $\forall (\mathbf{x}_{\compset_\rr},\mathbf{x}_{\compset_\rr}',y_{\compset_\rr})\in {\compset_\rr}$, $(\mathbf{x}_{\compset_\rr}, y_{\compset_\rr})\in \infset_{\rr}$ and $(\mathbf{x}_{\compset_\rr}', y_{\compset_\rr})\in \infset_{\rr}$
\end{thm}
\begin{proof}
    Let $(\mathbf{x}_{\compset_\rr},\mathbf{x}_{\compset_\rr}',y_{\compset_\rr})\in \compset_\rr$ and $i\in\{0,1,\dots,n\}$. By the definition of $\compset_\rr$, $\mathbf{x}$ is $(1^i0^{n-i},1^{i+1}0^{n-(i+1)},0^n)$ or $(1^{i}0^{n-i},1^{i-1}0^{n-(i-1)},0^n)$. Therefore, $(\mathbf{x}_{\compset_\rr}, y_{\compset_\rr})\in \infset_{\rr}$ and $(\mathbf{x}_{\compset_\rr}', y_{\compset_\rr})\in \infset_{\rr}$ as desired.
\end{proof}

By the previous theorems, the privacy bound algorithm is correct for $\infset_{\rr}$ and $\compset_{\rr}$. 

\subsubsection{Complexity}
Since $|\infset_{\rr}|=n+1$ and $|\compset_{\rr}|=2n$, the inference algorithm (Algorithm \ref{alg:inf}) runs in $O(\wmctime\cdot n)$ time and the privacy bound algorithm (Algorithm \ref{alg:priv}) runs in $O(\wmctime\cdot n)$ time as well.

% ACCURACY

\subsection{Computing Accuracy}
For accuracy computation, we consider the counting query on \rr{}. Here, we take the output vector $\mathbf{y}_i$ of \rr{} and compute the number of $1$'s in that vector. In this case, the input space is $\inputspace=\{0,1\}^n$ and output space is $\{0,1,\dots,n\}$. We refer to this counting version of \rr{} as \rrcount. Here, the set $\targetvals$ is equivalent to the counting function over Booleans, $count:\{0,1\}^n\rightarrow\{0,\dots,n\}$.

Here too we have a key lemma about the symmetries.
\begin{lemma}
    \label{lem:accsym}
    Given $i\in\{0,\dots,n\}$, $\forall y\in\{0,\dots,n\}$, for any $\mathbf{x}$ such that $count(\mathbf{x})=i$, it is true that $\pr{\rrcount(\mathbf{x})=y}=\pr{\rrcount(1^i0^{n-i})=y}$.
\end{lemma}
\begin{proof}
    Let $\mathbf{x}\in\inputspace$, $y\in\outputspace$ such that $count(\mathbf{x})=i$. By the definition of $\rrcount$, $\pr{\rrcount(\mathbf{x})=y}=\sum_{y_{\rr}\text{ s.t. } count(y_{\rr})=y}\pr{\rr(x)=y_{\rr}}=$ $\pr{\rrcount(1^i0^{n-i})=y}$.
\end{proof}

\subsubsection{Constructing $\accset$} We define the accuracy set to be 
\begin{equation}
    \accset_{\rrcount}=\{1^i0^{n-i}\}_{i\in \{0,\dots,n\}}
\end{equation}

\begin{thm}
    $\accset_{\rrcount}$ is an accuracy set for $\rrcount$.
\end{thm}
\begin{proof}
    Let $\mathbf{x}\in  \inputspace$ and $count(\mathbf{x})=i$. By Lemma \ref{lem:accsym}, $\pr{\rrcount(\mathbf{x})=y}=\pr{\rrcount(1^i0^{n-i})=y}$. Since $count(\mathbf{x})=count(1^i0^{n-i})$, $\sum_{y\in[count(\mathbf{x})-\alpha,count(\mathbf{x})+\alpha]}\pr{\rrcount(\mathbf{x})=y}=$\\
    $\sum_{y\in[count(1^i0^{n-i})-\alpha,count(1^i0^{n-i})+\alpha]}\pr{\rrcount(\mathbf{x}_\accset)=y}.$

    By definition, $\forall\mathbf{x}\in \accset_{\rrcount}$, $\mathbf{x}\in \inputspace$. Therefore $\compset_{\rrcount}$ is a valid accuracy set.
\end{proof}

We define the inference set to accommodate all necessary values.
\begin{equation}
    \infset_{\rrcount}=\bigcup_{i\in[0,n]}\bigcup_{j\in[i-\alpha,i+\alpha]}\{(1^i0^{n-i},j)\}
\end{equation}

\begin{thm}
    $\infset_{\rrcount}$ contains all $(\mathbf{x},y)$ pairs such that $\mathbf{x}\in\accset_{\rrcount}$ and $y\in[\mathcal{V}[\mathbf{x}]-\alpha,\mathcal{V}[\mathbf{x}]+\alpha]$
\end{thm}
\begin{proof}
    Let $\mathbf{x}\in\accset_{\rrcount}$ such that $count(\mathbf{x})=i$. By definition of $\accset_{\rrcount}$, $\mathbf{x}=1^i0^{n-i}$. By definition of $\infset_{\rrcount}$, $\forall y\in [i-\alpha,i+\alpha]$, $(1^i0^{n-i},y)$ as desired.
\end{proof}

By these theorems, the accuracy bound algorithm with this inference and accuracy set is correct.

\subsubsection{Complexity}
We have $|\infset_{\rrcount}|=n\cdot2\alpha$ and $|\accset_{\rrcount}|=(n+1)\cdot \alpha$. Therefore inference (Algorithm \ref{alg:inf}) and accuracy bound synthesis (Algorithm \ref{alg:acc}) both run in $O(\wmctime\cdot n\alpha)$ (or $O(\wmctime\cdot n^2)$ time since $2\alpha\leq n$).

\section{Evaluation}
\label{sec:eval}

We have shown that our method is theoretically sound. In this section we demonstrate how our implemented method far outperforms prior methods even for the exhaustive case. We also show how the expressiveness of probabilistic programming can be leveraged to easily implement new algorithms with more complex input and output spaces, a feat which would be difficult using prior techniques for exact verification. We then go on to discuss how our method computes accuracy and demonstrate the importance of also having access to accuracy analysis when evaluating private algorithms.

We implement the weighted model counting solution with BDDs by running programs in the Dice probabilistic programming language \cite{holtzen2020scaling}. We compare with a Markov chain model checking solution using the Storm probabilistic model checker \cite{hensel2021storm}. All experiments are written in Python and use the solver (Dice or Storm) as a subprocess. We implement $\mathsf{INFERENCE}$ (Algorithm \ref{alg:inf}), the privacy bound synthesis algorithm (Algorithm \ref{alg:priv}), and the accuracy bound synthesis algorithm (Algorithm \ref{alg:acc}) in Python. To easily compare performance on different parameters, we also develop a tool to generate Dice programs and Storm models for different DP algorithms, input/output spaces, and randomization parameters. Experiments are run on a 16 core AMD EPYC with 64GB of RAM and all code will be made publicly available.

\subsection{Comparison to Markov Chain Model Checking}

The heart of our method is using weighted model counting on tractable circuits to vastly improve both the size of the models and the implementation runtime. In our experiments, we use the Dice probabilistic programming language to model and compute tight privacy bounds for the randomized response algorithm as explained in Section \ref{sec:rr} for variable number of clients $n$. 

Prior work on inference for exact DP verification uses model checking on Markov chains to perform exact inference \cite{liu2018model,liu2022verifying}. The examples in these papers are limited and require adaptation to solve our synthesis problem. We compare our solution with a similar Markov chain model checking technique and confirm that our WMC solution outperforms the Markov chain solution in model size, number of solver runs, and inference time.

In Table \ref{tab:storm}, we see that, for various $n$, Dice develops an extremely compact BDD for randomized response, where the model size is $n+2$. This contrasts with the Markov chain solution where the models are size $7^n$. This is because every state in the Markov chain represents a possible setting of every free variable in the model. With each client added in randomized response, we multiply the number of free variables in the Markov chain, which causes an exponential increase in the number of states. These tools were developed to handle dynamic programs for small state spaces, and are therefore not intended to handle problems with large numbers of free variables and little to no notion of dynamics over times. BDDs, however, are exactly optimized for this kind of problem.
\begin{table*}
    \setlength\arrayrulewidth{1pt}
    \centering
    \resizebox{\textwidth}{!}{%
    \begin{tabular}{|l|l|l|l|l|l|l|l|l|}
    \hline
    \textbf{Method}  & \textbf{n} & \textbf{BDD Size} & \textbf{\# States} & \textbf{\# Trans.} & \textbf{\# Solver Runs} & \textbf{Inf. Time (s)} & \textbf{Synth. Time (s)} & \textbf{Build Time (s)} \\ \hline
    Storm Exhaustive & 2 & - & 49 & 100 & 16 & 0.0104 & $<$ 0.0001 & 0.0001 \\ \specialrule{.05pt}{0pt}{0pt}
    Storm Exhaustive & 4 & - & 2401 & 8488 & 256 & 0.8877 & 0.0007 & 0.0005 \\ \specialrule{.05pt}{0pt}{0pt}
    Storm Exhaustive & 6 & - & TO & TO & TO & TO & TO & TO \\ \hline
    Dice Exhaustive & 2 & 4 & - & - & 4 & 0.152 & 0.0001 & 0.0005 \\ \specialrule{.05pt}{0pt}{0pt}
    Dice Exhaustive & 4 & 6 & - & - & 16 & 0.4128 & 0.0011 & 0.0015 \\ \specialrule{.05pt}{0pt}{0pt}
    Dice Exhaustive & 6 & 8 & - & - & 64 & 1.5683 & 0.0226 & 0.008 \\ \specialrule{.05pt}{0pt}{0pt}
    Dice Exhaustive & 8 & 10 & - & - & 256 & 8.097 & 0.5285 & 0.024 \\ \specialrule{.05pt}{0pt}{0pt}
    Dice Exhaustive & 10 & 12 & - & - & 1024 & 63.1988 & 11.016 & 0.0917 \\ \hline
    Storm Restricted & 5 & - & 16807 & 73054 & 32 & 1.371 & 0.0001 & 0.0001 \\ \specialrule{.05pt}{0pt}{0pt}
    Storm Restricted & 10 & - & TO & TO & TO & TO & TO & TO \\ \hline
    Dice Restricted & 5 & 7 & - & - & 1 & 0.0239 & $<$ 0.0001 & 0.0001 \\ \specialrule{.05pt}{0pt}{0pt}
    Dice Restricted & 10 & 12 & - & - & 1 & 0.0518 & 0.0001 & 0.0001 \\ \specialrule{.05pt}{0pt}{0pt}
    Dice Restricted & 15 & 17 & - & - & 1 & 1.2864 & 0.0001 & 0.0001 \\ \specialrule{.05pt}{0pt}{0pt}
    Dice Restricted & 20 & 22 & - & - & 1 & 53.9282 & 0.0002 & 0.0001 \\ \hline
    \end{tabular}%
    }
    \caption{Comparison of state space and runtimes for a BDD weighted model counting solution (Dice) and a Markov chain model checking solution (Storm). Model size for Dice is the number of nodes in the BDD (BDD size) and for Storm is the number of states and transitions in the Markov chain. TO means that the experiment timed out on probabilistic inference. For Dice and Storm we provide results for both the exhaustive solution and restricted solution which uses symmetry sets as described in Section \ref{sec:rr}. Model size is independent of the symmetry sets, e.g. the model size for $n=10$ is $12$ for both Dice exhaustive and Dice restricted. \# Solver Runs indicates the number of calls to the solver required to compute the necessary probabilities. Inf. and Synth. time are the the total times in seconds for running implementations of the Inference and Privacy Bound Synthesis algorithms including subprocess times. Build time is the time it takes to build the program or model sent to the solver.}
    \label{tab:storm}
\end{table*}
The Markov chain solution also requires a manually crafted Markov chain which means that ours might not be the optimal encoding of this algorithm. However, finding an optimal encoding is a hard problem and requires technical knowledge on the internals of the model checking process. A benefit of our solution is that the Dice program is almost identical to the randomized response pseudocode, and the optimization of the BDD size is handled automatically. 

Another benefit of inference via WMC on BDDs demonstrated in Table \ref{tab:storm} is that the Dice solution is able to compute the entire output distribution of the program for a given input in one run, whereas the Markov chain model checking solution requires a run for each input/output pair. As shown in Table \ref{tab:storm} we have to run Storm $4^n$ times in the exhaustive case, but even in the exhaustive case we only have to run Dice $2^n$ times. We leverage the symmetries in the randomized response algorithm such that for the Dice restricted solution, we only have to run the solver once, compared with Storm where we still have to run it $2^n$ times. 

The outcome of these factors is that the Storm solution times out for $n=6$ in the exhaustive case and $n=10$ in the restricted case, while our solution with a Dice solver can find a tight bound in around a minute for $n=10$ in the exhaustive case, and less than a minute for $n=20$ in the restricted case as can be seen in Table \ref{tab:storm}.

\subsection{Expressiveness and the Geometric Above Threshold}
We have shown that our method outperforms Markov chain model checking for computing tight privacy bounds both for the exhaustive and restricted solutions in the case of binary randomized response. We now demonstrate the expressiveness of our solution using the truncated geometric above threshold mechanism. 

The geometric above threshold takes as input a length $n$ list of integers in $\{0,1, \dots,k\}$, and outputs an integer representing the index of the first value in the list that exceeds the threshold. We provide the complete pseudocode for the randomized (private) above threshold algorithm in Appendix \ref{apdx:above}. For methods like model checking with Markov chains, this would be a very difficult algorithm to encode as it has multiple invocations of the geometric mechanism and multiple growth dimensions which would require many explicit states and transitions to encode as a Markov chain.

We encode this algorithm as a simple $2n+3$ line Dice program and run our exhaustive privacy verification bound algorithm on it. We see in Fig. \ref{fig:above_bdds} that, even with the exhaustive solution, we are able to easily encode compact, computable BDDs for a variety of list lengths and max int values.

Not only does Dice provide a simple encoding, but we see that the optimized BDDs are extremely reasonably sized. Even for integers between $0$ and $3$ and list lengths up to $6$ (which means $6^4=1296$ possible inputs and two invocations of the geometric distribution over $4$ values each), we see that the generated BDD is able to represent this with only $91$ nodes. This is still less than the number of states in the Markov chain for the much simpler randomized response algorithm for $n=3$ when looking at the Storm solution. This means that exact inference, even for more complex DP algorithms can be orders of magnitude more efficient than previously believed.

\begin{figure}[t]
    \centering
    \includegraphics[width=.3\linewidth]{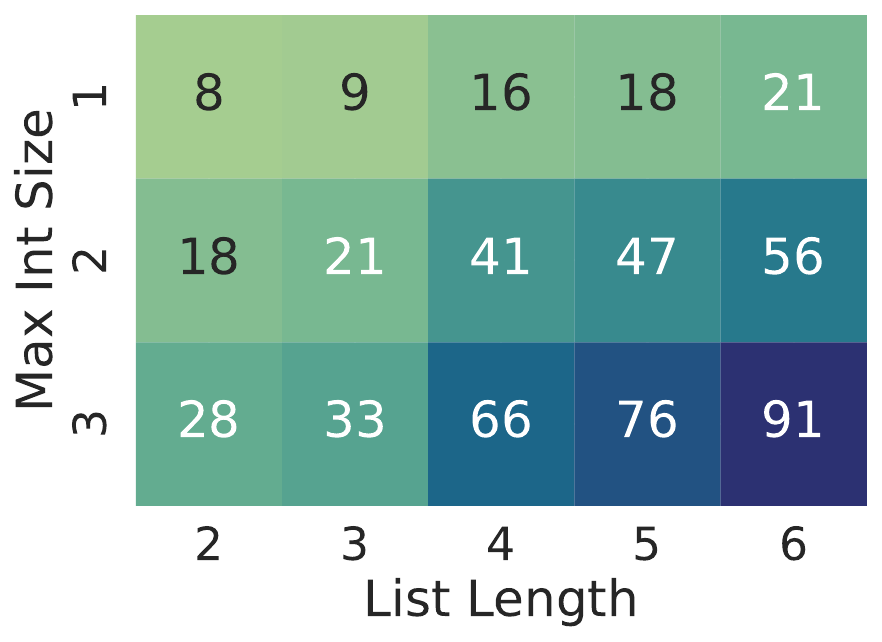}
    \caption{BDD sizes for the truncated geometric above threshold algorithm for various maximum integer sizes ($k$) and list lengths ($n$).}
    \label{fig:above_bdds}
\end{figure}

\subsection{Accuracy}
A key benefit of our method is its ability to compute tight accuracy bounds. We see in Fig. \ref{fig:rr_bdds} that having large integral output space (for example an output space of $\{1\dots n+1\}$ as is the case for the randomized response algorithm with a counting wrapper) causes the BDD sizes to grow more quickly than in privacy computation. As we noted in the theoretical portion of the paper, finding accuracy bounds also has larger time complexity due to the range of outputs bounded by $\alpha$. However we are still able to compute accuracy bounds for small $n$ in under 4 minutes, as shown in Fig. \ref{fig:rr_times}. 

\begin{figure}
    \centering
    \begin{subfigure}[t]{.28\linewidth}
        \centering
        \includegraphics[width=\linewidth]{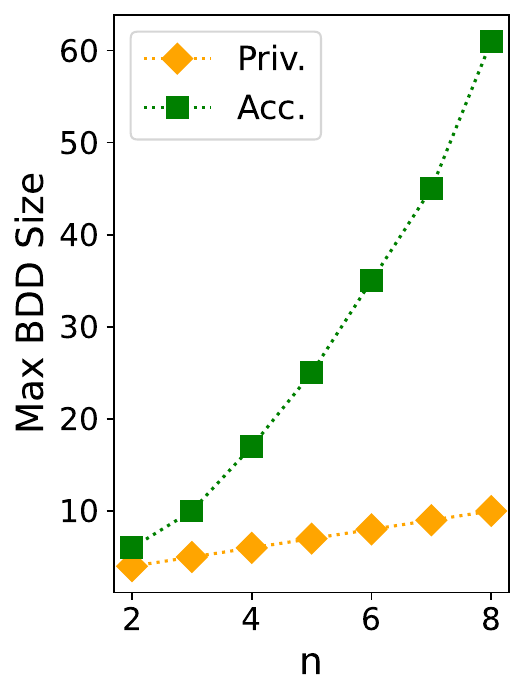}
        \caption{BDD Sizes}
        \label{fig:rr_bdds}
    \end{subfigure}
    \qquad
    \begin{subfigure}[t]{.363\linewidth}
        \centering
        \includegraphics[width=\linewidth]{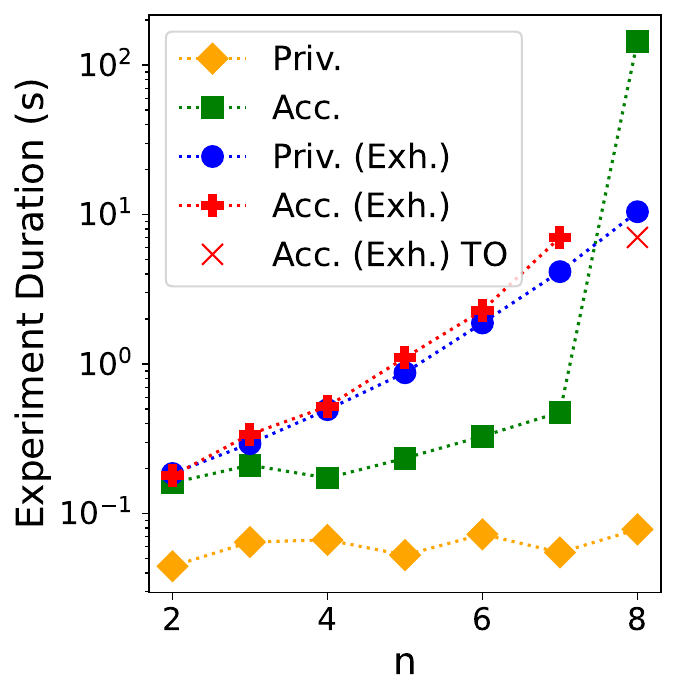}
        \caption{Experiment Durations}
        \label{fig:rr_times}
    \end{subfigure}
    \caption{BDD sizes and experiment durations for accuracy bound synthesis for randomized response with counting over various numbers of clients $n$ and $\lambda=0.2$.}
\end{figure}

Even though we are only able to perform tight accuracy bound synthesis on small examples, we begin to see the utility of our program as a tool in the algorithm design process. For example, in Fig. \ref{fig:bounds} we use our tool to find tight privacy and accuracy bounds for noise parameters. We see that noise is not linearly related to privacy or accuracy, as often implicitly assumed when algorithm designers minimize noise as a proxy for accuracy, as is done for differentially private synthesis as in \cite{roy2021learning}. Using such a proxy algorithm will result in an inexact outcome, and it is important to have a more detailed view of accuracy, even if $n$ is small.

\begin{figure}
    \centering
    \begin{subfigure}[t]{.35\linewidth}
        \centering
        \includegraphics[width=\linewidth]{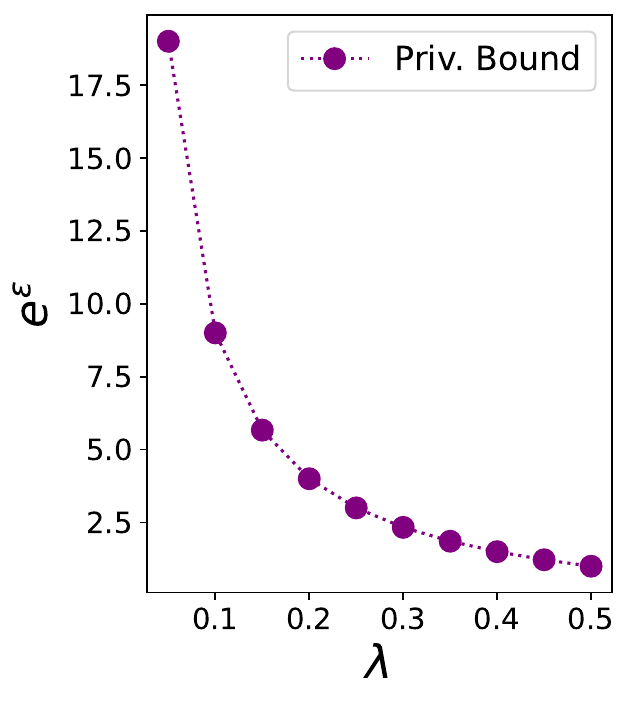}
        \caption{Privacy Bound (RR)}
        \label{fig:e_to_eps}
    \end{subfigure}
    \qquad
    \begin{subfigure}[t]{.35\linewidth}
        \centering
        \includegraphics[width=\linewidth]{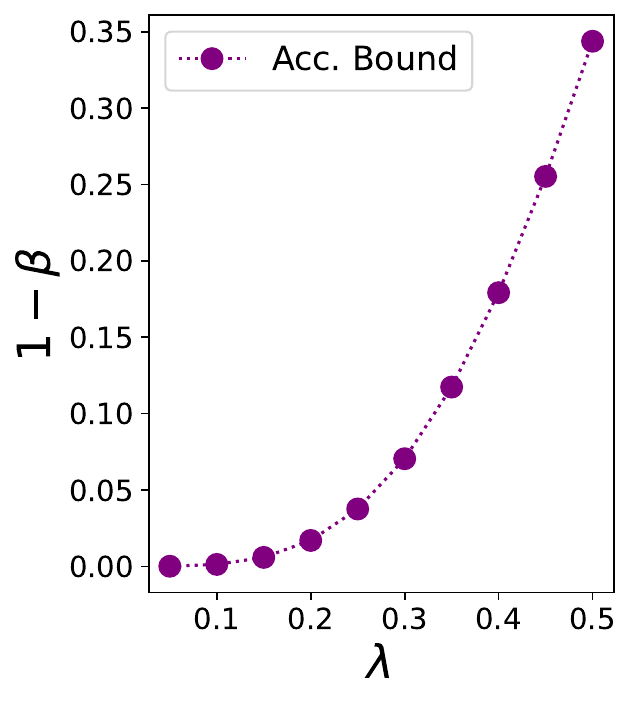}
        \caption{Accuracy (RR counting)}
        \label{fig:1-beta}
    \end{subfigure}
    \caption{Tight Accuracy and privacy bounds over varying coin flip parameters, $\lambda$, for binary randomized response with $n=8$. Generated using the Dice restricted privacy and accuracy bound synthesis implementation.}
    \label{fig:bounds}
\end{figure}

Our method goes beyond bound synthesis. We can run a variation of the accuracy bound synthesis algorithm and get each example for which there is a unique $1-\beta$, sorted by their values. We show a result generated by our tool of 4 inputs for which $1-\beta$ is lowest for randomized response with $n=8$ and $\lambda=0.2$ in Table \ref{tab:worst}. An algorithm designer can use this information to inform which examples have the greatest impact on accuracy.

\begin{table}[]
    \centering
    \resizebox{.23\linewidth}{!}{%
    \begin{tabular}{l|l}
     \textbf{Input} & $\mathbf{1-\beta}$  \\ \hline
     00000000 & 0.9437184 \\
     11111111 & 0.9437184 \\
     01111111 & 0.9723904 \\
     00000001 & 0.9723904
    \end{tabular}%
    }
    \caption{4 inputs for which $1-\beta$ is lowest for randomized response counting with $n=8$ and $\lambda=0.2$. Generated using the Dice restricted accuracy bound synthesis implementation.}
    \label{tab:worst}
\end{table}

\section{Related Works}
\label{sec:related}
There are a number of related works on differential privacy verification. In \cite{liu2018model,liu2022verifying}, Liu \etal{} provide a probabilistic model checking approach for exactly solving the differential privacy verification problem. We provide a comparison with this approach in Section \ref{sec:eval}.

Other works which are closely related to ours are works on finding counterexamples to differential privacy or tightening bounds. In \cite{ding2018detecting}, Ding \etal{} develop a statistical procedure for finding a candidate counterexample to differential privacy in probabilistic algorithms. In \cite{wang2020checkdp}, Wang \etal{} develop a program analysis tool for finding counterexamples or proofs of validation of privacy in implementations of DP algorithms. In \cite{bichsel2018dp}, Bichsel \etal{} use a sampling approach to find counterexamples to differential privacy. In \cite{zhang2020testing}, Zhang \etal{} use interpreters to find counterexamples to differential privacy.

There is significant work on developing type systems and program logics for differential privacy algorithms. This work begins with \cite{reed2010distance} in which Reed \etal{} develop a type system for differential privacy. Following this, \cite{gaboardi2013linear,barthe2015higher,zhang2017lightdp,near2019duet,fredrikson2014satisfiability} present other type systems and program logics for differential privacy. There is a related line of work which develops and implements a hoare logic approach and uses probabilistic couplings to mechanize proofs of differential privacy via the apRHL logic and EasyCrypt \cite{barthe2013verified,barthe2012probabilistic,barthe2014proving,barthe2016proving,hsu2017probabilistic}. In \cite{albarghouthi2017synthesizing}, Albarghouthi \etal{} propose a technique to automate these proofs. In \cite{barthe2016differentially}, Barthe \etal{} develop a probabilistic programming framework which uses inference for writing verifiable DP programs.

There are also works on program synthesis for DP programs \cite{roy2021learning,smith2019synthesizing,wang2021dpgen}. Our techniques could be used in these synthesis approaches for verifying intermediate programs, and as a more exact method of ensuring high accuracy while satisfying DP properties.

In a more theoretical approach, Gaboardi \etal{} \cite{gaboardi2019complexity} find complexity bounds for exact verification of DP in non-looping programs. In follow-up work to this, Bun \etal{} \cite{bun2022complexity} extend this analysis to looping programs.

Privacy auditing provides empirical methods to estimate the privacy leakage of an ML algorithm by mounting privacy attacks. Privacy auditing can be performed with membership inference attacks~\cite{pmlr-v202-zanella-beguelin23a} or data poisoning~\cite{AuditingDP,nasr2021adversary}, but initial techniques developed for auditing require training thousands of ML models to provide confidence intervals for the estimated lower bound on the privacy parameter. Recent techniques showed how to rigorously perform estimation of the privacy parameter while using multiple randomized canaries~\cite{pillutla2023unleashing} and eventually training a single ML model~\cite{andrew2023oneshot,steinke2023auditing}.

To our knowledge, there is very little work on theoretical or applied formal verification or bound synthesis for accuracy of DP algorithms. In \cite{barthe2021deciding}, Barthe \etal{} provide theoretical analysis on the decidability of accuracy different classes of probabilistic computations. In \cite{lobo2020programming}, Lobo-Vesga \etal{} develop a programming framework for estimating accuracy bounds.

\section{Conclusion and Future Work}
\label{sec:conclusion}
We believe that our novel approach for synthesizing tight privacy and accuracy bounds has a lot of promise for developing techniques for computing accuracy and privacy properties in a wide range of differential privacy applications. Our approach uses state-of-the-art techniques for probabilistic inference via weighted model counting, and can benefit from ongoing advances in artificial intelligence and automated reasoning. 

Though we are able to vastly improve automation from prior work by leveraging expressive probabilistic programming languages, one limitation is the manual step of identifying symmetries to define the inference, comparison, and accuracy sets. While there is room for improvement, in our experiments we demonstrate how the power of our WMC solution provides advances and utility, even for the exhaustive approach. In future work, we plan to investigate techniques from SAT solving and symmetry breaking to improve automation. 

In future work, we can also extend this framework to other definitions of differential privacy including Renyi and approximate differential privacy \cite{dwork2014algorithmic}. Additionally, we can look at more complex systems with more complicated probability distributions such as randomized response with amplification by shuffling \cite{cheu2019distributed}.

Our work opens the door for further automation and accessibility of using probabilistic programming languages techniques for verification and synthesis of DP algorithms using state-of-the-art inference tools.

\section*{Acknowledgment}
We would like to thank Professors Marco Gaboardi and Jon Ullman for their technical guidance and input on differential privacy case studies. Thanks also to Lydia Zakynthinou, Konstantina Bairaktari, Ludmila Glinskih, Minsung Cho, and LaKyah Tyner for discussions about approaches and theoretical analysis. This work has been supported by NSF grant CNS-2247484.
\bibliographystyle{IEEEtran}
\bibliography{refs.bib}

\begin{thebibliography}{10}
\providecommand{\url}[1]{#1}
\csname url@samestyle\endcsname
\providecommand{\newblock}{\relax}
\providecommand{\bibinfo}[2]{#2}
\providecommand{\BIBentrySTDinterwordspacing}{\spaceskip=0pt\relax}
\providecommand{\BIBentryALTinterwordstretchfactor}{4}
\providecommand{\BIBentryALTinterwordspacing}{\spaceskip=\fontdimen2\font plus
\BIBentryALTinterwordstretchfactor\fontdimen3\font minus
  \fontdimen4\font\relax}
\providecommand{\BIBforeignlanguage}[2]{{%
\expandafter\ifx\csname l@#1\endcsname\relax
\typeout{** WARNING: IEEEtran.bst: No hyphenation pattern has been}%
\typeout{** loaded for the language `#1'. Using the pattern for}%
\typeout{** the default language instead.}%
\else
\language=\csname l@#1\endcsname
\fi
#2}}
\providecommand{\BIBdecl}{\relax}
\BIBdecl

\bibitem{DworkMNS16}
C.~Dwork, F.~McSherry, K.~Nissim, and A.~Smith, ``Calibrating noise to
  sensitivity in private data analysis,'' \emph{Journal of Privacy and
  Confidentiality}, vol.~7, no.~3, 2016.

\bibitem{mironov2012significance}
I.~Mironov, ``On significance of the least significant bits for differential
  privacy,'' in \emph{Proceedings of the 2012 ACM conference on Computer and
  communications security}, 2012, pp. 650--661.

\bibitem{tramer2022debugging}
F.~Tramer, A.~Terzis, T.~Steinke, S.~Song, M.~Jagielski, and N.~Carlini,
  ``Debugging differential privacy: A case study for privacy auditing,''
  \emph{arXiv preprint arXiv:2202.12219}, 2022.

\bibitem{stevens2022backpropagation}
T.~Stevens, I.~C. Ngong, D.~Darais, C.~Hirsch, D.~Slater, and J.~P. Near,
  ``Backpropagation clipping for deep learning with differential privacy,''
  \emph{arXiv preprint arXiv:2202.05089}, 2022.

\bibitem{barthe2013verified}
G.~Barthe, G.~Danezis, B.~Gr{\'e}goire, C.~Kunz, and S.~Zanella-Beguelin,
  ``Verified computational differential privacy with applications to smart
  metering,'' in \emph{2013 IEEE 26th Computer Security Foundations
  Symposium}.\hskip 1em plus 0.5em minus 0.4em\relax IEEE, 2013, pp. 287--301.

\bibitem{barthe2012probabilistic}
G.~Barthe, B.~K{\"o}pf, F.~Olmedo, and S.~Zanella~Beguelin, ``Probabilistic
  relational reasoning for differential privacy,'' in \emph{Proceedings of the
  39th annual ACM SIGPLAN-SIGACT symposium on Principles of programming
  languages}, 2012, pp. 97--110.

\bibitem{barthe2014proving}
G.~Barthe, M.~Gaboardi, E.~J.~G. Arias, J.~Hsu, C.~Kunz, and P.-Y. Strub,
  ``Proving differential privacy in hoare logic,'' in \emph{2014 IEEE 27th
  Computer Security Foundations Symposium}.\hskip 1em plus 0.5em minus
  0.4em\relax IEEE, 2014, pp. 411--424.

\bibitem{barthe2016proving}
G.~Barthe, M.~Gaboardi, B.~Gr{\'e}goire, J.~Hsu, and P.-Y. Strub, ``Proving
  differential privacy via probabilistic couplings,'' in \emph{Proceedings of
  the 31st Annual ACM/IEEE Symposium on Logic in Computer Science}, 2016, pp.
  749--758.

\bibitem{hsu2017probabilistic}
J.~Hsu, \emph{Probabilistic couplings for probabilistic reasoning}.\hskip 1em
  plus 0.5em minus 0.4em\relax University of Pennsylvania, 2017.

\bibitem{ding2018detecting}
Z.~Ding, Y.~Wang, G.~Wang, D.~Zhang, and D.~Kifer, ``Detecting violations of
  differential privacy,'' in \emph{Proceedings of the 2018 ACM SIGSAC
  Conference on Computer and Communications Security}, 2018, pp. 475--489.

\bibitem{wang2020checkdp}
Y.~Wang, Z.~Ding, D.~Kifer, and D.~Zhang, ``Checkdp: An automated and
  integrated approach for proving differential privacy or finding precise
  counterexamples,'' in \emph{Proceedings of the 2020 ACM SIGSAC Conference on
  Computer and Communications Security}, 2020, pp. 919--938.

\bibitem{bichsel2018dp}
B.~Bichsel, T.~Gehr, D.~Drachsler-Cohen, P.~Tsankov, and M.~Vechev,
  ``Dp-finder: Finding differential privacy violations by sampling and
  optimization,'' in \emph{Proceedings of the 2018 ACM SIGSAC Conference on
  Computer and Communications Security}, 2018, pp. 508--524.

\bibitem{zhang2020testing}
H.~Zhang, E.~Roth, A.~Haeberlen, B.~C. Pierce, and A.~Roth, ``Testing
  differential privacy with dual interpreters,'' \emph{Proceedings of the ACM
  on Programming Languages}, vol.~4, no. OOPSLA, pp. 1--26, 2020.

\bibitem{AuditingDP}
\BIBentryALTinterwordspacing
M.~Jagielski, J.~Ullman, and A.~Oprea, ``Auditing differentially private
  machine learning: How private is private {SGD?}'' in \emph{Proceedings of
  Advances in Neural Information Processing Systems}, ser. NeurIPS, vol.~33,
  2020, pp. 22\,205--22\,216. [Online]. Available:
  \url{https://proceedings.neurips.cc/paper/2020/file/fc4ddc15f9f4b4b06ef7844d6bb53abf-Paper.pdf}
\BIBentrySTDinterwordspacing

\bibitem{nasr2021adversary}
\BIBentryALTinterwordspacing
M.~Nasr, S.~Song, A.~Thakurta, N.~Papernot, and N.~Carlini, ``Adversary
  instantiation: Lower bounds for differentially private machine learning,'' in
  \emph{42nd {IEEE} Symposium on Security and Privacy, {SP} 2021, San
  Francisco, CA, USA, 24-27 May 2021}.\hskip 1em plus 0.5em minus 0.4em\relax
  {IEEE}, 2021, pp. 866--882. [Online]. Available:
  \url{https://doi.org/10.1109/SP40001.2021.00069}
\BIBentrySTDinterwordspacing

\bibitem{andrew2023oneshot}
G.~Andrew, P.~Kairouz, S.~Oh, A.~Oprea, H.~B. McMahan, and V.~Suriyakumar,
  ``One-shot empirical privacy estimation for federated learning,''
  \emph{CoRR}, vol. abs/2302.03098, 2023.

\bibitem{nasr2023tight}
M.~Nasr, J.~Hayes, T.~Steinke, B.~Balle, F.~Tram\`{e}r, M.~Jagielski,
  N.~Carlini, and A.~Terzis, ``Tight auditing of differentially private machine
  learning,'' in \emph{Proceedings of the 32nd USENIX Conference on Security
  Symposium}, ser. SEC '23.\hskip 1em plus 0.5em minus 0.4em\relax USA: USENIX
  Association, 2023.

\bibitem{pillutla2023unleashing}
\BIBentryALTinterwordspacing
K.~Pillutla, G.~Andrew, P.~Kairouz, H.~B. McMahan, A.~Oprea, and S.~Oh,
  ``Unleashing the power of randomization in auditing differentially private
  {ML},'' in \emph{Thirty-seventh Conference on Neural Information Processing
  Systems}, 2023. [Online]. Available:
  \url{https://openreview.net/forum?id=mlbes5TAAg}
\BIBentrySTDinterwordspacing

\bibitem{steinke2023auditing}
\BIBentryALTinterwordspacing
M.~J. Thomas~Steinke, Milad~Nasr, ``Privacy auditing with one (1) training
  run,'' in \emph{Thirty-seventh Conference on Neural Information Processing
  Systems}, 2023. [Online]. Available:
  \url{https://openreview.net/forum?id=mlbes5TAAg}
\BIBentrySTDinterwordspacing

\bibitem{roy2021learning}
S.~Roy, J.~Hsu, and A.~Albarghouthi, ``Learning differentially private
  mechanisms,'' in \emph{2021 IEEE Symposium on Security and Privacy
  (SP)}.\hskip 1em plus 0.5em minus 0.4em\relax IEEE, 2021, pp. 852--865.

\bibitem{barthe2020decidingpriv}
G.~Barthe, R.~Chadha, V.~Jagannath, A.~P. Sistla, and M.~Viswanathan,
  ``Deciding differential privacy for programs with finite inputs and
  outputs,'' in \emph{Proceedings of the 35th Annual ACM/IEEE Symposium on
  Logic in Computer Science}, 2020, pp. 141--154.

\bibitem{liu2018model}
D.~Liu, B.-Y. Wang, and L.~Zhang, ``Model checking differentially private
  properties,'' in \emph{Asian Symposium on Programming Languages and
  Systems}.\hskip 1em plus 0.5em minus 0.4em\relax Springer, 2018, pp.
  394--414.

\bibitem{liu2022verifying}
------, ``Verifying pufferfish privacy in hidden markov models,'' in
  \emph{International Conference on Verification, Model Checking, and Abstract
  Interpretation}.\hskip 1em plus 0.5em minus 0.4em\relax Springer, 2022, pp.
  174--196.

\bibitem{chavira2008probabilistic}
M.~Chavira and A.~Darwiche, ``On probabilistic inference by weighted model
  counting,'' \emph{Artificial Intelligence}, vol. 172, no. 6-7, pp. 772--799,
  2008.

\bibitem{holtzen2020scaling}
S.~Holtzen, G.~Van~den Broeck, and T.~Millstein, ``Scaling exact inference for
  discrete probabilistic programs,'' \emph{Proceedings of the ACM on
  Programming Languages}, vol.~4, no. OOPSLA, pp. 1--31, 2020.

\bibitem{holtzen2020generating}
S.~Holtzen, T.~Millstein, and G.~Van~den Broeck, ``Generating and sampling
  orbits for lifted probabilistic inference,'' in \emph{Uncertainty in
  Artificial Intelligence}.\hskip 1em plus 0.5em minus 0.4em\relax PMLR, 2020,
  pp. 985--994.

\bibitem{dwork2014algorithmic}
C.~Dwork, A.~Roth \emph{et~al.}, ``The algorithmic foundations of differential
  privacy,'' \emph{Foundations and Trends{\textregistered} in Theoretical
  Computer Science}, vol.~9, no. 3--4, pp. 211--407, 2014.

\bibitem{HoltzenCAV21}
S.~Holtzen, S.~Junges, M.~Vazquez-Chanlatte, T.~Millstein, S.~A. Seshia, and
  G.~{Van den Broeck}, ``Model checking finite-horizon markov chains with
  probabilistic inference,'' in \emph{Proceedings of the 33rd International
  Conference on Computer-Aided Verification (CAV)}, July 2021.

\bibitem{bichsel2021dp}
B.~Bichsel, S.~Steffen, I.~Bogunovic, and M.~Vechev, ``Dp-sniper: Black-box
  discovery of differential privacy violations using classifiers,'' in
  \emph{2021 IEEE Symposium on Security and Privacy (SP)}.\hskip 1em plus 0.5em
  minus 0.4em\relax IEEE, 2021, pp. 391--409.

\bibitem{van2021introduction}
G.~Van~den Broeck, K.~Kersting, S.~Natarajan, and D.~Poole, \emph{An
  Introduction to Lifted Probabilistic Inference}.\hskip 1em plus 0.5em minus
  0.4em\relax MIT Press, 2021.

\bibitem{sabharwal2009symchaff}
A.~Sabharwal, ``Symchaff: exploiting symmetry in a structure-aware
  satisfiability solver,'' \emph{Constraints}, vol.~14, pp. 478--505, 2009.

\bibitem{aloul2006efficient}
F.~A. Aloul, K.~A. Sakallah, and I.~L. Markov, ``Efficient symmetry breaking
  for boolean satisfiability,'' \emph{IEEE Transactions on Computers}, vol.~55,
  no.~5, pp. 549--558, 2006.

\bibitem{hensel2021storm}
C.~Hensel, S.~Junges, J.-P. Katoen, T.~Quatmann, and M.~Volk, ``The
  probabilistic model checker storm,'' \emph{International Journal on Software
  Tools for Technology Transfer}, pp. 1--22, 2021.

\bibitem{reed2010distance}
J.~Reed and B.~C. Pierce, ``Distance makes the types grow stronger: a calculus
  for differential privacy,'' in \emph{Proceedings of the 15th ACM SIGPLAN
  international conference on Functional programming}, 2010, pp. 157--168.

\bibitem{gaboardi2013linear}
M.~Gaboardi, A.~Haeberlen, J.~Hsu, A.~Narayan, and B.~C. Pierce, ``Linear
  dependent types for differential privacy,'' in \emph{Proceedings of the 40th
  annual ACM SIGPLAN-SIGACT symposium on Principles of programming languages},
  2013, pp. 357--370.

\bibitem{barthe2015higher}
G.~Barthe, M.~Gaboardi, E.~J. Gallego~Arias, J.~Hsu, A.~Roth, and P.-Y. Strub,
  ``Higher-order approximate relational refinement types for mechanism design
  and differential privacy,'' \emph{ACM SIGPLAN Notices}, vol.~50, no.~1, pp.
  55--68, 2015.

\bibitem{zhang2017lightdp}
D.~Zhang and D.~Kifer, ``Lightdp: Towards automating differential privacy
  proofs,'' in \emph{Proceedings of the 44th ACM SIGPLAN Symposium on
  Principles of Programming Languages}, 2017, pp. 888--901.

\bibitem{near2019duet}
J.~P. Near, D.~Darais, C.~Abuah, T.~Stevens, P.~Gaddamadugu, L.~Wang,
  N.~Somani, M.~Zhang, N.~Sharma, A.~Shan \emph{et~al.}, ``Duet: an expressive
  higher-order language and linear type system for statically enforcing
  differential privacy,'' \emph{Proceedings of the ACM on Programming
  Languages}, vol.~3, no. OOPSLA, pp. 1--30, 2019.

\bibitem{fredrikson2014satisfiability}
M.~Fredrikson and S.~Jha, ``Satisfiability modulo counting: A new approach for
  analyzing privacy properties,'' in \emph{Proceedings of the Joint Meeting of
  the Twenty-Third EACSL Annual Conference on Computer Science Logic (CSL) and
  the Twenty-Ninth Annual ACM/IEEE Symposium on Logic in Computer Science
  (LICS)}, 2014, pp. 1--10.

\bibitem{albarghouthi2017synthesizing}
A.~Albarghouthi and J.~Hsu, ``Synthesizing coupling proofs of differential
  privacy,'' \emph{Proceedings of the ACM on Programming Languages}, vol.~2,
  no. POPL, pp. 1--30, 2017.

\bibitem{barthe2016differentially}
G.~Barthe, G.~P. Farina, M.~Gaboardi, E.~J.~G. Arias, A.~Gordon, J.~Hsu, and
  P.-Y. Strub, ``Differentially private bayesian programming,'' in
  \emph{Proceedings of the 2016 ACM SIGSAC Conference on Computer and
  Communications Security}, 2016, pp. 68--79.

\bibitem{smith2019synthesizing}
C.~Smith and A.~Albarghouthi, ``Synthesizing differentially private programs,''
  \emph{Proceedings of the ACM on Programming Languages}, vol.~3, no. ICFP, pp.
  1--29, 2019.

\bibitem{wang2021dpgen}
Y.~Wang, Z.~Ding, Y.~Xiao, D.~Kifer, and D.~Zhang, ``Dpgen: Automated program
  synthesis for differential privacy,'' in \emph{Proceedings of the 2021 ACM
  SIGSAC Conference on Computer and Communications Security}, 2021, pp.
  393--411.

\bibitem{gaboardi2019complexity}
M.~Gaboardi, K.~Nissim, and D.~Purser, ``The complexity of verifying loop-free
  programs as differentially private,'' \emph{arXiv preprint arXiv:1911.03272},
  2019.

\bibitem{bun2022complexity}
M.~Bun, M.~Gaboardi, and L.~Glinskih, ``The complexity of verifying boolean
  programs as differentially private,'' in \emph{2022 IEEE 35th Computer
  Security Foundations Symposium (CSF)}.\hskip 1em plus 0.5em minus 0.4em\relax
  IEEE, 2022, pp. 396--411.

\bibitem{pmlr-v202-zanella-beguelin23a}
\BIBentryALTinterwordspacing
S.~Zanella-Beguelin, L.~Wutschitz, S.~Tople, A.~Salem, V.~R\"{u}hle, A.~Paverd,
  M.~Naseri, B.~K\"{o}pf, and D.~Jones, ``{B}ayesian estimation of differential
  privacy,'' in \emph{Proceedings of the 40th International Conference on
  Machine Learning}, ser. Proceedings of Machine Learning Research, A.~Krause,
  E.~Brunskill, K.~Cho, B.~Engelhardt, S.~Sabato, and J.~Scarlett, Eds., vol.
  202.\hskip 1em plus 0.5em minus 0.4em\relax PMLR, 23--29 Jul 2023, pp.
  40\,624--40\,636. [Online]. Available:
  \url{https://proceedings.mlr.press/v202/zanella-beguelin23a.html}
\BIBentrySTDinterwordspacing

\bibitem{barthe2021deciding}
G.~Barthe, R.~Chadha, P.~Krogmeier, A.~P. Sistla, and M.~Viswanathan,
  ``Deciding accuracy of differential privacy schemes,'' \emph{Proceedings of
  the ACM on Programming Languages}, vol.~5, no. POPL, pp. 1--30, 2021.

\bibitem{lobo2020programming}
E.~Lobo-Vesga, A.~Russo, and M.~Gaboardi, ``A programming framework for
  differential privacy with accuracy concentration bounds,'' in \emph{2020 IEEE
  Symposium on Security and Privacy (SP)}.\hskip 1em plus 0.5em minus
  0.4em\relax IEEE, 2020, pp. 411--428.

\bibitem{cheu2019distributed}
A.~Cheu, A.~Smith, J.~Ullman, D.~Zeber, and M.~Zhilyaev, ``Distributed
  differential privacy via shuffling,'' in \emph{Advances in
  Cryptology--EUROCRYPT 2019: 38th Annual International Conference on the
  Theory and Applications of Cryptographic Techniques, Darmstadt, Germany, May
  19--23, 2019, Proceedings, Part I 38}.\hskip 1em plus 0.5em minus 0.4em\relax
  Springer, 2019, pp. 375--403.

\bibitem{balcer2017differential}
V.~Balcer and S.~Vadhan, ``Differential privacy on finite computers,''
  \emph{arXiv preprint arXiv:1709.05396}, 2017.

\end{thebibliography}

\appendix
\section{Truncated Geometric Above Threshold}
\label{apdx:above}
The above threshold algorithm takes as input an array of values and a thresholding value $T$ and outputs the index of the first value in the input array that exceeds $T$. The DP version of this algorithm adds noise both to each entry and to the threshold $T$ \cite{dwork2014algorithmic}. 
\begin{algorithm}
    \caption{Truncated Geometric Above Threshold}
    \label{alg:above}
    \begin{algorithmic}[1]
    \renewcommand{\algorithmicrequire}{\textbf{Input:}}
    \renewcommand{\algorithmicensure}{\textbf{Output:}}
    \REQUIRE Array of input integers $X$, threshold value $T$, randomness parameters $\lambda_1$ and $\lambda_2$
    \ENSURE Index $i$
    \STATE $\hat{T}\gets T+Geom(\lambda_1)$
    \FOR{$x_i\in X$}
        \STATE $r_i\gets Geom(\lambda_2)$
        \IF{$x_i+r_i\geq \hat{T}$}
            \STATE return $i$
        \ENDIF
    \ENDFOR
    \RETURN $\mathbf{y}$
    \end{algorithmic}
\end{algorithm}

We use a discrete, finite variation of this DP algorithm which is defined over integer-valued input arrays, integer-valued $T$, and uses the truncated geometric mechanism for the added noise \cite{balcer2017differential}. We provide the pseudocode for this variation in Algorithm \ref{alg:above}. In our experiments, we set randomness parameters $\lambda_1=\lambda_2$ for simplicity.

\end{document}